\documentclass[acmsmall,screen,nonacm]{acmart}

\usepackage[english]{babel}
\usepackage[utf8x]{inputenc}
\usepackage[T1]{fontenc}
\usepackage{forest}
\usepackage{tikz-qtree}
\usepackage{tikz}
\usetikzlibrary{calc}
\usepackage{amsmath}
 
\usepackage{amssymb}
\usepackage{threeparttable}
\usepackage{booktabs}
\usepackage[normalem]{ulem}
\usepackage{enumitem}
\usepackage{multirow}
\usepackage{algorithm}
\usepackage[noend]{algorithmic}

\usepackage{graphicx}
\usepackage[colorinlistoftodos]{todonotes}
\usepackage{varwidth}
\usepackage{capt-of}
\usepackage{color}
\usepackage{subfig}
\usepackage{blindtext}
\usepackage{wrapfig,lipsum,booktabs}

\newcommand{\eat}[1]{}
\newcommand{\A}[1]{#1} 

\newcommand{\E}[1]{\emph{#1}}
\newcommand{\R}[1]{\textbf{#1}}

\newcommand{\emptySet}{\brac{}}
\newcommand{\brac}[1]{\{#1\}}
\newcommand{\blue}[1]{\textcolor{black}{#1}}

\newcommand{\cyan}[1]{\textcolor{cyan}{#1}}

\newcommand{\red}[1]{\textcolor{red}{#1}}
\newcommand{\fei}[1]{\noindent \red{Comment(Fei): #1}}

\newcommand{\reviseOne}[1]{\noindent \textcolor{black}{#1}}

\newcommand{\reviseTwo}[1]{\noindent \textcolor{black}{#1}}
\newcommand{\reviseThree}[1]{\noindent \textcolor{black}{#1}}

\newcommand{\rev}[1]{\noindent \textcolor{purple}{#1}}

\newcommand{\lambdaphi}{$\Lambda({\phi})$}
\newcommand{\Omegalambda}{$\Omega(\lambda)$}
\newcommand{\rhoomega}{$\rho_{\Omega, \lambda}$}
\newcommand{\cost}{$\mathcal{C}$}

\newcommand{\eatTR}[1]{#1} 
\newcommand{\eatNTR}[1]{}  

\newcommand{\set}[1]{#1}

\newcommand{\tbf}{\textbf{\textcolor{red}{X}}\xspace}
\newcommand{\lit}[1]{{\sl #1}}

\newtheorem{example}{Example}[section]
\newtheorem{definition}{Definition}
\newtheorem{theorem}{Theorem}[section]
\newtheorem{lemma}[theorem]{Lemma}

\newcommand{\ofdclean}{${\sf OFDClean}$\xspace}
\newcommand{\fastofd}{${\sf FastOFD}$\xspace}

\AtBeginDocument{%
  \providecommand\BibTeX{{%
    \normalfont B\kern-0.5em{\scshape i\kern-0.25em b}\kern-0.8em\TeX}}}

\setcopyright{acmcopyright}
\copyrightyear{2022}
\acmYear{2022}
\acmDOI{10.1145/3524303}
\acmJournal{JDIQ}
\acmVolume{14}
\acmNumber{2}

\acmMonth{3}

\begin{document}

\fancyhead{}
\fancyfoot{}

\title{Discovery and Contextual Data Cleaning with Ontology Functional Dependencies}

\author{Zheng Zheng}
\email{zhengz13@mcmaster.ca}
\affiliation{%
  \institution{McMaster University}
    \country{Canada}
}
\author{Longtao Zheng}
\authornote{Work done as an intern at McMaster University.}
\email{zlt0116@mail.ustc.edu.cn}
\affiliation{%
  \institution{University of Science and Technology of China}
  \country{China}
}
\author{Morteza Alipourlangouri}
\email{alipoum@mcmaster.ca}
\affiliation{%
  \institution{McMaster University}
    \country{Canada}
}
\author{Fei Chiang}
\email{fchiang@mcmaster.ca}
\affiliation{%
  \institution{McMaster University}
      \country{Canada}
}
\author{Lukasz Golab}
\email{lgolab@uwaterloo.ca}
\affiliation{%
  \institution{University of Waterloo}
      \country{Canada}
}
\author{Jaroslaw~Szlichta}
\email{jarek@ontariotechu.ca}
\affiliation{%
  \institution{Ontario Tech University}
      \country{Canada}
}

\begin{abstract}
Functional Dependencies (FDs) define attribute relationships based on syntactic equality, and, when used in data cleaning, they erroneously label syntactically different but semantically equivalent values as errors. We explore dependency-based data cleaning with \emph{Ontology Functional Dependencies} (OFDs), which express semantic attribute relationships such as synonyms \eat{and is-a hierarchies}defined by an ontology.  
We study the theoretical foundations of OFDs, including sound and complete axioms and a linear-time inference procedure.  We then propose an algorithm for discovering OFDs (exact ones and ones that hold with some exceptions) from data that uses the axioms to prune the search space.
Towards enabling OFDs as data quality rules in practice,  we study the problem of finding minimal repairs to a relation and ontology with respect to a set of OFDs.
We demonstrate the effectiveness of our techniques on real datasets, and show that OFDs can significantly reduce the number of false positive errors in data cleaning techniques that rely on traditional FDs.
\end{abstract}

%

\maketitle
\section{Introduction}\label{sec:introduction}
In constraint-based data cleaning, dependencies are used to specify data quality requirements.  Data that are inconsistent with respect to the dependencies are identified as erroneous, and updates to the data are generated to re-align the data with the dependencies.  
Existing approaches use Functional Dependencies (FDs) \cite{BFFR05,PSC15}, Inclusion Dependencies \cite{BFFR05}, Conditional Functional Dependencies \cite{CFGJM07}, Denial Constraints \cite{CIP13}, \reviseTwo{Order Dependencies~\cite{Szl2012}, and Matching Dependencies~\cite{FanLMTY11}} to define the attribute relationships that the data must satisfy.  However, these approaches are limited to identifying attribute relationships based on syntactic equivalence (or syntactic similarity for Metric FDs \cite{KSS09,PSC15}), and unable to convey semantic equivalence, which is often necessary in data cleaning.

\begin{table*}
\scriptsize
	\begin{minipage}{0.4\linewidth}
		\centering
		\caption{Sample clinical trials.}
        \begin{tabular}{ | l | l | l | l | l | l | l | l |}
        \hline
        \hspace{-2 mm} \textbf{id} \hspace{-2 mm} & \textbf{CC}    & \textbf{CTRY} & \hspace{-2 mm} \textbf{SYMP} & \hspace{-2 mm} \textbf{TEST}  & \hspace{-2 mm} \textbf{DIAG}  & \hspace{-2 mm} \textbf{MED}  \\
        \hline \hline
        \hspace{-2 mm} $t_1$ \hspace{-2 mm} &      US & USA &  joint pain & CT & osteoarthritis & ibuprofen \\
        \hspace{-2 mm} $t_2$ \hspace{-2 mm} &      IN & India & joint pain  & CT  & osteoarthritis & NSAID  \\
        \hspace{-2 mm} $t_3$ \hspace{-2 mm} &     CA & Canada & joint pain  & CT  & osteoarthritis & naproxen \\
        \hspace{-2 mm} $t_4$ \hspace{-2 mm} &      IN & Bharat & nausea  & EEG  & migrane &  analgesic  \\
        \hspace{-2 mm} $t_5$ \hspace{-2 mm} &       US & America & nausea  & EEG  & migrane &  tylenol \\
        \hspace{-2 mm} $t_6$ \hspace{-2 mm} &      US & USA &  nausea  &  EEG  & migrane &  acetaminophen \\
        \hspace{-2 mm} $t_7$ \hspace{-2 mm} & IN & India & chest pain  &  X-ray  & hypertension & morphine \\
        \hspace{-2 mm} $t_8$ \hspace{-2 mm} & US & USA & headache & CT & hypertension & cartia \\
		\hspace{-2 mm} $t_9$ \hspace{-2 mm} & US & USA & headache & MRI & hypertension & tiazac \textcolor{blue}{(ASA)} \\
		\hspace{-2 mm} $t_{10}$ \hspace{-2 mm} & US & America & headache & MRI & hypertension & tiazac \\
		\hspace{-2 mm} $t_{11}$ \hspace{-2 mm} & US & USA & headache & CT & hypertension & tiazac \textcolor{blue}{(adizem)}\\
        \hline 
        \end{tabular}
		\label{tab:cleanexample}
	\end{minipage}\hfill
	\begin{minipage}{0.4\linewidth}
	\centering
		\includegraphics[width=65mm]{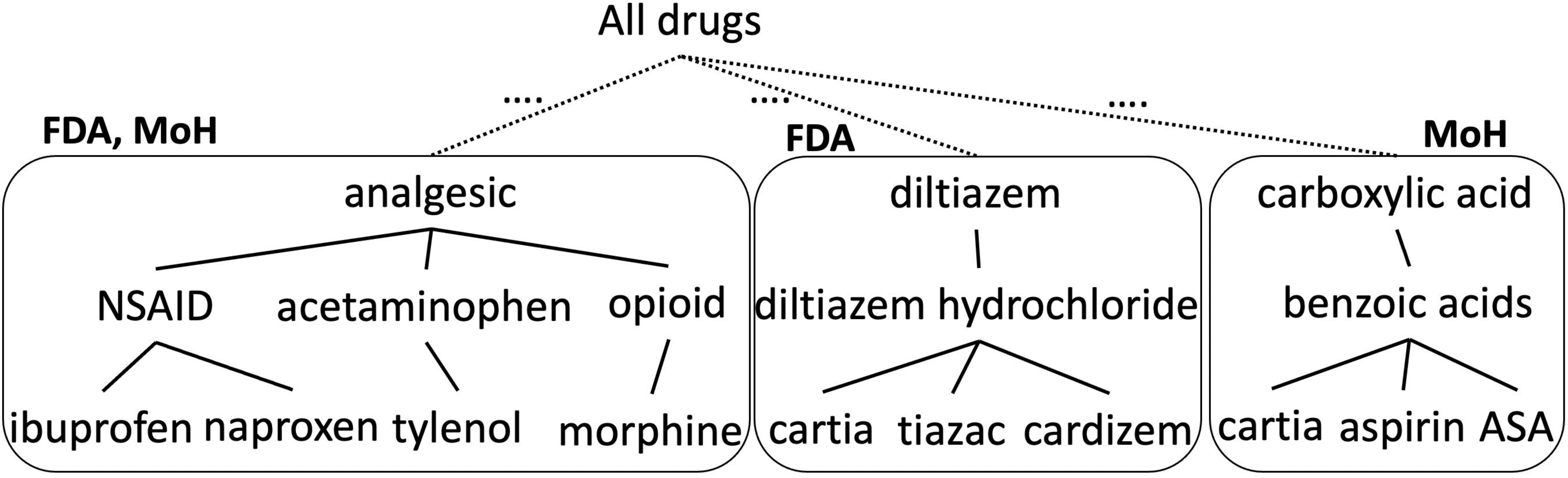}
	\vskip -0.3cm
		\captionof{figure}{Medical drug ontology.}  
		\label{fig:ontology}
	\end{minipage}
	\vspace{-0.5cm}
\end{table*}

\begin{example}
\label{ex:intro}
Table \ref{tab:cleanexample} shows a sample of clinical trial records containing patient country codes (CC), country (CTRY), symptoms (SYMP), diagnosis (DIAG), and the prescribed medication (MED). Consider two FDs: $F_{1}$:  [CC]  $\rightarrow$  [CTRY] and $F_{2}$: [SYMP, DIAG] $\rightarrow$ [MED]. \eat{and $F_{3}$: [MED] $\rightarrow$ [DIAG]. \fei{Unfortunately, this is not a very realistic FD F3.}} The tuples ($t_{1}, t_{5}$, $t_{6}$, $t_{8}-t_{11}$) do not satisfy $F_{1}$ as  \lit{America} and \lit{USA} are not syntactically equivalent (the same is true for ($t_{2}, t_{4}$, $t_{7}$)).   However, \lit{USA} is synonymous with \lit{America}, and ($t_{1}, t_{5}$, $t_{6}$, $t_{8}-t_{11}$) all refer to the same country.  Similarly, \lit{Bharat} in $t_{4}$ is synonymous with \lit{India} as it is the country's original Sanskrit name. For $F_{2}$, ($t_{1}-t_{3}$), ($t_{4}-t_{6}$) and ($t_{8}-t_{11}$) are violations as the consequent values all refer to different medications. However, as shown in Figure \ref{fig:ontology}, with domain knowledge from a medical ontology~\cite{Dron}, we see that the values participate in an inheritance relationship. Both \lit{ibuprofen} and \lit{naproxen} are non-steroidal anti-inflammatory drugs (\lit{NSAID}), \lit{tylenol} is an \lit{acetaminophen} drug, which in turn is an \lit{analgesic}, and both \lit{cartia} and \lit{tiazac} are \lit{diltiazem hydrochloride}. 
\end{example}

\eat{\fei{What about morphine in $t_7$?} Similarly, tuples $t_{9}-t_{11}$ violate $F_{3}$ since \lit{hyperpiesia} is not equivalent to \lit{hypertension}. However, with additional information from a disease diagnosis  ontology~\cite{DOID}, \lit{hyperpiesia} and \lit{hypertension} are indeed the same disease. \fei{These two values are not in Figure 2 ontology.  The disease diagnosis ontology needs to be expanded to be more realistic.} }

The above example demonstrates that real data contain domain-specific relationships beyond syntactic equivalence or similarity.  It also highlights two common relationships that occur between two values $u$ and $v$: (1) $u$ and $v$ are \emph{synonyms}; and (2) $u$ is-a $v$ denoting \emph{inheritance}.  These relationships are often defined within domain-specific ontologies.
Unfortunately, traditional FDs and their extensions are unable to capture these relationships, and existing data cleaning approaches flag tuples containing synonymous and inheritance values as errors.  This leads to an increased number of reported ``errors'', and a larger search space of data repairs to consider.

To address these shortcomings, our earlier work proposed a novel class of dependencies called \emph{Ontology Functional Dependencies} (OFDs) that capture synonyms and is-a relationships defined in an ontology~\cite{baskaran2017efficient}.  In this paper, we focus on synonyms, which are commonly used in practice, and we study cleaning a relation and an ontology with respect to (w.r.t.) a set of synonym OFDs.  
What makes OFDs interesting is the notion of \emph{senses}, which determine how a dependency should be interpreted for a given ontology; e.g., \lit{jaguar} can be interpreted as an \lit{animal} or as a \lit{vehicle}, \reviseTwo{country codes vary according to multiple standards (interpretations) such as the International Standards Organization (ISO) vs. United Nations (UN).} 
To make OFDs useful in practice, where data semantics are often poorly documented and change, we propose an algorithm to discover OFDs.  

OFDs serve as contextual data quality rules that enforce the semantics modeled in an ontology.  However, application requirements change, data evolve, and as new knowledge is generated, inconsistencies arise between the data, OFDs, and ontologies.   For example, the US Food and Drug Administration (FDA) has a monthly approval cycle to certify new drugs. If downstream data applications do not update their data and ontologies, this leads to stale and missing values in the ontology, and inconsistencies w.r.t. the data and the OFDs~\cite{Dron}.  Similarly, changes to the data may be required to re-align the data to the OFDs and the ontology.  Consider the following example.

\begin{example}\label{ex:motivating}
Consider Table~\ref{tab:cleanexample} with updated values (shown in blue) in $t_{9}$[MED] and $t_{11}$[MED] to reflect changes in a patient's prescribed medication.  Tuples ($t_{8} - t_{11}$) are now inconsistent w.r.t. the previously defined $F_2$.  If we augment $F_{2}$ with additional semantics, provided by a medication ontology as shown in Figure~\ref{fig:ontology}, tuple $t_{11}$ continues to be problematic since \lit{adizem} is not defined in the ontology, and is not equivalent to \lit{cartia} nor \lit{tiazac}. The updated value in $t_9$[MED] to \lit{ASA} leads to an inconsistency since \lit{ASA} is not semantically equivalent to \lit{cartia} nor \lit{tiazac}. We must find an interpretation of the ontology, called a \emph{sense} (denoted in bold in Figure~\ref{fig:ontology}), where the values \{\lit{ASA}, \lit{cartia}, \lit{tiazac}, \lit{adizem}\} are all equivalent.  Unfortunately, there is no sense in which all these values are semantically equivalent.  \eat{The ontology shows there is no sense in which all these We see that \lit{cartia} and \lit{tiazac} represent the same medication, but \lit{cartia} and \lit{ASA} are equivalent drugs by the Ministry of Health (MoH) in Israel.}  To resolve these violations, we can: (1) repair the ontology by adding the value \lit{adizem} and \lit{ASA} under the FDA sense; or (2)  update tuples ($t_8 - t_{11}$) to either \lit{cartia} or \lit{ASA} under the MoH sense.  In both cases, there now exists a sense where all MED values in ($t_8 - t_{11}$) are equivalent.
\end{example}

The above example demonstrates that there are multiple options to resolve violations, namely, modifying the ontology or the data. 
\eat{Defining approximate OFDs, which relax the notion of satisfaction to hold over a subset of the relation, can be used to identify error values and candidate fixes.  
}  We study the repair problem that re-aligns the data, dependencies (OFDs), and an ontology via a minimal number of repairs to the data and to the ontology.  State-of-the-art data cleaning solutions have taken a holistic approach to combine attribute relationships with external knowledge bases and statistical methods to determine the most likely repair \cite{PSC15}, and probabilistic models that learn from clean distributions in the data to minimally repair dirty values~\cite{YBE13}.  \reviseTwo{However, none of these techniques consider repairs to ontologies.  We adopt a similar philosophy to prior work that consider repairs to the data and/or constraints~\cite{CM11,BIGG13}.}  We make the following contributions~\footnote{\reviseTwo{We note that the first three contributions are published in our earlier work~\cite{baskaran2017efficient}.}} :

\begin{enumerate}

\item We define OFDs based on synonym relationships.  In contrast to existing dependencies, OFDs include attribute relationships that go beyond syntactic equality or similarity, and consider the notion of \emph{senses} that provide the interpretations under which the dependencies are evaluated.  In contrast to FDs, OFDs are not amenable to tuple-to-tuple comparisons and instead they must be verified over equivalence classes of tuples. 

\item We introduce a \emph{sound and complete} set of axioms for OFDs.  While the inference complexity of other FD extensions is co-NP complete, we show that  inference for OFDs remains linear.

\item
We propose the \fastofd algorithm to discover a \emph{complete} and \emph{minimal} set of OFDs to alleviate the burden of specifying them manually.  We show that OFDs can be discovered efficiently by traversing a \emph{set-containment lattice} with \emph{exponential} worst-case complexity in the number of attributes, the same as for traditional FDs, and polynomial complexity in the number of tuples.  We develop pruning rules based on our axiomatization.

\item We present \ofdclean, an algorithm that computes a minimal number of modifications to an ontology and to the data to satisfy a given set of OFDs. To consider the possible interpretations of the data, \ofdclean selects the best interpretation (sense) for an equivalence class of tuples, such that the selected sense minimizes the number of required modifications.  


%
\item We evaluate the performance and effectiveness of \fastofd and \ofdclean using real datasets over varying parameter values.  

\end{enumerate}

We give preliminary definitions in Section~\ref{sec:prelim}.  We present our axiomatization and inference procedure in Section~\ref{sec:foundations}, and  Section~\ref{sec:fastofd} describes our discovery algorithm.  We introduce data and ontology repairs, and the \ofdclean framework in Section~\ref{sec:cleaning}.  In Section~\ref{sec:sense}, we describe our sense selection algorithm, and then present our ontology repair algorithm in Section~\ref{sec:ontrep}.   We present experimental results in Section~\ref{sec:eval}, related work in Section~\ref{sec:rw}, and conclude in Section~\ref{sec:conclusion}.


\section{Preliminaries and Definitions}
\label{sec:prelim}


\textbf{Functional Dependencies.} 
A functional dependency (FD) $F$ over a relation $R$ is represented as $X \rightarrow A$, where $X$ is a set of attributes and $A$ is a single attribute in $R$.  An instance $I$ of $R$ satisfies $F$ if for every pair of tuples $t_1, t_2 \in I$, if $t_1$[$X$] = $t_2$[$X$], then $t_1$[$A$] = $t_2$[$A$].  A partition of $X$, $\Pi_{X}$, is the set of equivalence classes containing tuples with equal values in $X$.  Let $x_i$ be an equivalence class with a representative $i$ that is equal to the smallest tuple id in the class, and $|x_i|$ be the size of the equivalence class (in some definitions, we drop the subscript and refer to an equivalence class simply as $x$).  For example, in Table \ref{tab:cleanexample}, $\Pi_{CC}$ = \{\{$t_1,t_5,t_6$\}\{$t_2,t_4,t_7$\}\{$t_3$\}\}.  

\textbf{Sense.} 
An ontology $S$ is an explicit specification of a domain that includes concepts, entities, properties, and relationships among them.  \reviseTwo{In this work, we consider tree-shaped ontologies for simplicity.} These constructs are often defined and applicable only under a given interpretation, called a \emph{sense}.   The meaning of these constructs for a given $S$ can be modeled according to different senses leading to different ontological interpretations.  As mentioned previously, the value \lit{jaguar} can be interpreted under two senses: (1) as an animal, and (2) as a vehicle.  As an animal, \lit{jaguar} is synonymous with \lit{panthera onca}, but not with the value \lit{jaguar land rover automotive}.

We define classes $E$ for the interpretations or senses defined in $S$.
Let \emph{synonyms}$(E)$ be the set of all synonyms for a given class $E$. For instance, $synonyms(E1)$ = \{\lit{car, auto, vehicle}\}, $synonyms(E2)$ = \{\lit{jaguar, jaguar land rover}\} and $synonyms(E3)$ = \{\lit{jaguar, panthera onca}\}. Let \emph{names}$(C)$ be the set of all classes, i.e., interpretations or senses, for a given value $C$. For example, \emph{names}$(jaguar)$ = \{$E1$, $E2$, $E3$\} as jaguar can be an animal or a vehicle. Let $descendants(E)$ be a set of all representations for the class $E$ or any of its descendants, i.e., $descendants(E)$ = \{$s$ $|$ $s$ $\in$ \emph{synonyms}$(E)$  or $s$ $\in$ \emph{synonyms}$(E_{i})$, where $E_{i}$ \emph{is-a} $E_{i-1}$, ..., $E_{1}$ \emph{is-a} $E$\}, e.g., $descendants(E3)$ = \{\lit{jaguar, peruvian jaguar,mexican jaguar}\}.  

\textbf{Ontology Functional Dependencies}
We define OFDs w.r.t. a given ontology $S$.  We consider \reviseTwo{the synonym} relationship on the right-hand-side of a dependency leading to synonym OFDs. \eat{ and inheritance OFDs (Definition~\ref{def:inheritance_ofd}). }

\begin{definition} \label{def:synonym_ofd}
A relation instance $I$ satisfies a \emph{synonym OFD} $X \rightarrow_{syn} A$, if for each equivalence class $x \in \Pi_{X}(I)$, there exists a sense (interpretation) under which all the $A$-values of tuples in $x$ are synonyms.  That is, $X \rightarrow_{syn} A$ holds, if for each equivalence class $x$, $$\bigcap_{names(a), a \in \{t[A] | t \in x \}} \neq \emptyset.$$
\end{definition}

\begin{example}\label{example:syn}
Consider the OFD [CC]  $\rightarrow_{syn}$  [CTRY]  from Table \ref{tab:cleanexample}.    We have 
$\Pi_{CC}$ = \{\{$t_1,t_5,t_6$\}\{$t_2,t_4,t_7$\}\{$t_3$\}\}.  The first equivalence class, \{$t_1,t_5,t_6$\}, representing the value \lit{US}, corresponds to three distinct values of CTRY.  According to a geographical ontology, names(\lit{United States}) $\cap$ names(\lit{America}) $\cap$ names(\lit{USA}) = \lit{United States of America}.  
Similarly, the second class \{$t_2,t_4,t_7$\} gives names(\lit{India}) $\cap$ names(\lit{Bharat}) = \lit{India}.  The last equivalence class \{$t_3$\} contains a single tuple, so there is no conflict.  Since all references to CTRY in each equivalence class resolve to some common interpretation,  
the synonym OFD holds.
\end{example}

Synonym OFDs subsume traditional FDs, where all values are assumed to have a single literal interpretation (for all classes $E$, $|$\emph{synonyms}$(E)|$ = 1). 
Similar to traditional FDs, we say that an instance $I$ satisfies a set of OFDs $\Sigma$, denoted $I \models \Sigma$ \emph{if and only if} $I$ satisfies each OFD $\phi \in \Sigma$. Without loss of generality, we assume that the consequent (right side) of each $\phi \in \Sigma$ consists of a single attribute $A$ (this follows from the axioms that will be presented in Section~\ref{sec:foundations}).    When we refer to an OFD $\phi$: $X \rightarrow Y$, we say $I \models \phi$ iff $I \models (X \rightarrow A)$ for all $A \in Y$.  That is, for each attribute $A \in Y$, all the $A$-values of tuples in $x$ share a common sense (interpretation).
Ontological relationships are semantically meaningful on both sides of a dependency to capture a greater number of true positive errors. However, due to the large search space of ontological relationships on the left-hand side, since we must also consider all tuples outside of the partition group when checking for violations, we focus on the right-hand-side. This is similar to other seminal work~\cite{KSS09,PSC15} that considers small variations only on the right-hand-side via metric FDs.

\textbf{Relationship to other dependencies.}
The notion of senses makes OFDs non-trivial and interesting.  For each equivalence class, there must exist a common interpretation of the values. 
Checking pairs of tuples, as in traditional FDs (and Metric FDs \cite{KSS09,PSC15}, which assert that two tuples whose left-hand side attribute values are equal must have \emph{syntactically} similar right-hand side attribute values according to some distance metric), is not sufficient, as illustrated next.

Consider Table \ref{tab:mVio}, where the synonym OFD $X \rightarrow_{syn} Y$ does not hold (for each $Y$ value, we list its possible interpretations in the last column).  Although all pairs of $t[Y]$ values share a common class 
\begin{wraptable}{l}{3.2cm}
\vspace{-4mm}
\begin{threeparttable}
\scriptsize
\caption{Defining OFDs}
\label{tab:mVio}
\begin{tabular}{l|ccc}
\toprule
\textbf{id}	& \textbf{X} 	&  \textbf{Y}	&  \textbf{Classes for Y}  \cr
\midrule
$t_{1}$ &  \emph{u} &  \emph{v} &  \{\emph{C,D}\}	 \cr
$t_{2}$  &  \emph{u} & \emph{w} &   \{\emph{D,F}\} \cr 
$t_{3}$  &  \emph{u} &  \emph{z} &  \{\emph{C,F,G}\} \cr
\bottomrule
\end{tabular}
\end{threeparttable}
\vspace{-5mm}
\end{wraptable} 
(i.e., \{\emph{v,w}\}: \E{D}, \{\emph{v,z}\}: \E{C}, \{\emph{w,z}\}: \E{F}), the intersection of the classes is empty. 


Furthermore, OFDs \emph{cannot} be reduced to traditional FDs or Metric FDs.  Since values may have multiple senses,  
it is not possible to create a normalized relation instance by replacing each value with a unique canonical name.  Furthermore, ontological similarity is not a metric since it does not satisfy the identity of indiscernibles (e.g., for synonyms).
\section{Foundations and Optimizations}
\label{sec:foundations}

In this section, we provide a formal framework for OFDs.  We give a sound and complete axiomatization for OFDs that reveals how OFDs behave; notably, not all axioms that hold for traditional FDs carry over.  We then use the axioms to design pruning rules that will be used by our OFD discovery algorithm.
Finally, we provide a linear time inference procedure that ensures a set of OFDs remains minimal.   
\eatNTR{Due to space constraints, we defer all proofs to the accompanying technical report~\cite{BKCS16}.}

\subsection{Axiomatization for OFDs}
\label{sec:axioms}


We start with the \emph{closure} of a set of attributes $\set{X}$  over a set of OFDs $\Sigma$, which will allow us to determine whether additional OFDs follow from $\Sigma$ by axioms. We use the notation $\Sigma \vdash$ to state that $\set{X}$ $\rightarrow$ $\set{Y}$ is provable with axioms from $\Sigma$. 
\begin{definition} $(Closure)$ \label{def:closure}
The closure of $\set{X}$, denoted as $\set{X}^{+}$, with respect to the set of OFDs $\Sigma$ is defined as $\set{X}^{+}$ $=$ $\{\A{A}$ $|$ $\Sigma$ $\vdash$ $\set{X}$ $\rightarrow$ $\A{A} \}$.
\end{definition}

\begin{lemma}\label{lemma:closure}
$\Sigma$ $\vdash$ $\set{X}$ $\rightarrow$ $\set{Y}$ iff $\set{Y}$ $\subseteq$ $\set{X}^{+}$.
\end{lemma}

%

\begin{proof}
Let $\set{Y}$ $=$ $\{\A{A}_{1}$, $...$, $\A{A}_{n}\}$. Assume $\set{Y}$ $\subseteq$ $\set{X}^{+}$. By definition of $\set{X}^{+}$, $\set{X}$ $\rightarrow$ $\set{A}_{i}$, for all $i$ $\in$ $\{1,..., n\}$. By the Composition inference rule, $\set{X}$ $\rightarrow$ $\set{Y}$ follows. For the other direction, suppose $\set{X}$ $\rightarrow$ $\set{Y}$ follows from the axioms. For each $i$ $\in$ $\{1,..., n\}$, $\set{X}$ $\rightarrow$ $\A{A}_{i}$ follows by Decomposition, and $\set{Y}$
$\subseteq$ $\set{X}^{+}$.
\end{proof}

A sound and complete axiomatization for traditional FDs consists of Transitivity (if $X$ $\rightarrow$ $Y$ and $Y$ $\rightarrow$ $Z$ then $X$ $\rightarrow$ $Z$), Reflexivity (if $Y \subseteq X$ then $X$ $\rightarrow$ $Y$) and Composition. \reviseTwo{Since OFDs subsume traditional FDs, all lemmas and theorems for OFDs apply to FDs, but not vice versa.  For example, Transitivity does not hold for OFDs.} Consider a relation $R(A,B,C)$ with three tuples $\{(a, b, d), (a, c, e), (a, b, d)\}$.  Assume that $b$ is a synonym of $c$ and $d$ is not a synonym of $e$. The synonym OFD $\A{A}$ $\rightarrow_{syn}$ $\A{B}$ holds since $b$ and $c$ are synonyms.  In addition, $\A{B}$ $\rightarrow_{syn}$ $\A{C}$ holds as $b$ and $c$ are not equal.  However, the synonym OFD: $\A{A}$ $\rightarrow_{syn}$ $\A{C}$ does not hold as $d$ and $e$ are not synonyms

Theorem~\ref{theorem:complete} below presents a sound and complete set of axioms (inference rules) for OFDs. The Identity axiom generates \emph{trivial} dependencies that are always true. 
%


\begin{theorem}\label{theorem:complete}
These axioms are sound and complete for OFDs.
    \begin{enumerate}[nolistsep]
        \item[O1] Identity: $\forall X \subseteq R$, $X$ $\rightarrow$ $X$
        \item[O2] Decomposition: If $X$ $\rightarrow$ $Y$ and $Z$ $\subseteq$ $Y$
            then $X$ $\rightarrow$ $Z$
        \item[O3] Composition: If $X$ $\rightarrow$ $Y$ and $Z$ $\rightarrow$
        $W$ then $XZ$ $\rightarrow$ $YW$
    \end{enumerate}
\end{theorem}


\noindent{\scriptsize PROOF.}
First we prove that the axioms are sound. That is, if $\set{\Sigma}$ $\vdash$ $\set{X}$ $\rightarrow$ $\set{Y}$ then $\Sigma$ $\models$ $\set{X}$ $\rightarrow$ $\set{Y}$. The Identity axiom is clearly sound. We cannot have a relation with tuples that agree on $\set{X}$ yet are not in a synonym relationship. To prove Decomposition, suppose we have a relation that satisfies $\set{X}$ $\rightarrow$ $\set{Y}$ and $\set{Z} \subseteq \set{Y}$. Therefore, for all tuples that agree on $\set{X}$, they are in a synonym relationship on all attributes in $\set{Y}$ and hence, also on $\set{Z}$. Therefore, $\set{X}$ $\rightarrow$ $\set{Z}$. The soundness of Composition is an extension of the same argument.

Below we present the completeness proof, that is, if $\Sigma$ $\models$ $\set{X}$ $\rightarrow$ $\set{Y}$ then $\Sigma$ $\vdash$ $\set{X}$ $\rightarrow$ $\set{Y}$. Without loss of generality, we consider a table $I$ with three tuples shown in Table \ref{table:tablet}.  We divide the
\begin{wraptable}{l}{5.2cm}
\vspace{-3mm}

\begin{tabular}{|c|c|c|}
    \hline
    \multicolumn{2}{|c|}{$\set{X}^{+}$} & \multicolumn{1}{|c|}{}\\
    \hline
    $\set{X}$ & $\set{X}^{+} \setminus \set{X}$ & Other attributes\\
    \hline
    $v...v$ & $v...v$ & $v...v$\\
    \hline
    $v...v$ & $v'...v'$ & $w...w$\\
    \hline
\end{tabular}
\caption{Table template for OFDs.}\label{table:tablet}
\vskip -0.3cm

\vspace{-6mm}
\end{wraptable} 
attributes of $I$ into three subsets: $\set{X}$, the set consisting of attributes in the closure $\set{X}^{+}$ minus attributes in $\set{X}$, and all remaining attributes. Assume that the values $v$ and $v'$ are not equal ($v$ $\not =$ $v'$ and $v'$ $\not =$ $v''$), but they are in a synonym relationship. Also, $v$ and $w$ are not synonyms, and hence, they are also not equal.

We first show that all dependencies in the set of OFDs $\Sigma$ are satisfied by a table $I$ ($I$ $\models$ $F$). Since OFD axioms are sound, OFDs inferred from $\Sigma$ are true. Assume $\set{V}$ $\rightarrow$ $Z$ is in $\Sigma$, however, it is not satisfied by $I$. Therefore,  $\set{V}$ $\subseteq$ $\set{X}$ because otherwise the tuples of $I$ disagree on some attribute of $\set{V}$ since $v$ and $v'$  as well as $v$ and $w$ are not equal, and consequently an OFD $\set{V}$ $\rightarrow$ $\set{Z}$ would not be violated. Moreover, $\set{Z}$ cannot be a subset of $\set{X}^{r}$ ($\set{Z}$ $\not \subseteq$ $\set{X}^{+}$), or else $\set{V}$ $\rightarrow$ $\set{Z}$ would be satisfied by $I$. Let $\A{A}$ be an attribute of $\set{Z}$ not in $\set{X}^{+}$. Since, $\set{V}$ $\subseteq$ $\set{X}$, $\set{X}$ $\rightarrow$ $\set{V}$ by Reflexivity. Also a dependency $\set{V}$ $\rightarrow$ $\set{Z}$ is in $\Sigma$, hence, by Decomposition, $\set{V}$ $\rightarrow$ $\A{A}$. By Composition $\set{XV}$
$\rightarrow$ $\set{V}\A{A}$ can be inferred, therefore, $\set{X}$ $\rightarrow$ $\set{V}\A{A}$ as $\set{V}$ $\subseteq$ $\set{X}$. However, then Decomposition rule tells us that $\set{X}$ $\rightarrow$ $\A{A}$,
which would mean by the definition of the closure that $\A{A}$ is in $\set{X}^{+}$, which we assumed not to be the case. 
Contradiction. 

Our remaining proof obligation is to show that any OFD not inferable from a set of OFDs $\Sigma$ with OFD axioms ($\Sigma$ $\not \vdash$ $\set{X}$ $\rightarrow$ $\set{Y}$) is not true ($\Sigma$ $\not \models$ $\set{X}$ $\rightarrow$ $\set{Y}$). Suppose it is satisfied
($\Sigma$ $\models$ $\set{X}$ $\rightarrow$ $\set{Y}$). By Reflexivity, $\set{X}$ $\rightarrow$ $\set{X}$, therefore, by Lemma \ref{lemma:closure}, $\set{X}$ $\subseteq$ $X^{+}$. Since $\set{X}$ $\subseteq$ $\set{X}^{+}$, it follows by the construction of Table $I$ that $\set{Y}$ $\subseteq$
$\set{X}^{+}$. Otherwise, the tuples of Table $I$ agree on $\set{X}$ but are not in a synonym relationship on some attribute $\A{A}$ from $\set{Y}$. Then, from Lemma \ref{lemma:closure} it can be inferred that $\set{X}$ $\rightarrow$ $\set{Y}$. Contradiction. Thus, whenever $\set{X}$ $\rightarrow$ $\set{Y}$ does not follow from $\Sigma$ by OFDs
axioms, $\Sigma$ does not logically imply $\set{X}$ $\rightarrow$ $\set{Y}$. That is, the axiom system is complete, and this ends the proof of Theorem \ref{theorem:complete}. \qedsymbol


\eatTR{
\begin{lemma}$(Reflexivity)$\label{lemma:reflexivity}
If $\set{Y}$ $\subseteq$ $\set{X}$, then $\set{X}$ $\rightarrow$ $\set{Y}$.
\end{lemma}}

\eatTR{\begin{proof}
$\set{X}$ $\rightarrow$ $\set{X}$ holds by Identity axiom. Therefore, it can be inferred by the Decomposition inference rule that $\set{X}$ $\rightarrow$ $\set{Y}$ holds.
\end{proof}}

\eatTR{Union inference rule shows what can be inferred from two or more dependencies which have the same sets on the left side.}

\eatTR{\begin{lemma}$(Union)$ \label{lemma:union}
If $\set{X}$ $\rightarrow$ $\set{Y}$ and $\set{X}$ $\rightarrow$ $\set{Z}$, then $\set{X}$ $\rightarrow$ $\set{YZ}$.
\end{lemma}}

\eatTR{\begin{proof}
We are given $\set{X}$ $\rightarrow$ $\set{Y}$ and $\set{X}$ $\rightarrow$ $\set{Z}$. Hence, the Composition axiom can be used to infer $\set{X}$ $\rightarrow$ $\set{YZ}$.
\end{proof}}

\begin{definition}\emph{~\cite{Lien82}}
A functional dependency $X \rightarrow Y$ with \emph{nulls} called Null Functional Dependency (NFD) states that whenever two tuples agree on non-null values in $X$, they agree on the values in $Y$, which may be partial.
\end{definition}

\begin{theorem}\emph{~\cite{Lien82}}\label{theorem:NFDcomplete}
These axioms are sound and complete for NFDs.
    \begin{enumerate}[nolistsep]
        \item[N1] Reflexivity: $\forall Y \subseteq X$, $X$ $\rightarrow$ $Y$
        \item[N2] Append: If $X$ $\rightarrow$ $Y$ and $Z$ $\subseteq$ $W$, then $XW$ $\rightarrow$ $YZ$
        \item[N3] Union: If $X$ $\rightarrow$ $Y$ and $Y$ $\rightarrow$ $Z$, then $X$ $\rightarrow$ $Z$
        \item[N4] Simplification: If $X$ $\rightarrow$ $YZ$, then $X$ $\rightarrow$ $Y$ and $X$ $\rightarrow$ $Z$
    \end{enumerate}
\end{theorem}

Interestingly, the definitions of OFDs and NFDs are semantically different, i.e., a satisfying OFD does not necessarily imply a corresponding NFD is true (e.g., in Table~\ref{tab:cleanexample} an OFD [CC]  $\rightarrow$  [CTRY] holds, but a corresponding NFD [CC]  $\rightarrow$  [CTRY] does \emph{not} hold), and vice versa. Also, while data verification for FDs and NFDs can be done on pairs of tuples, for OFDs it has to be performed on an entire equivalence class over the left-hand-side attributes. Following Table~\ref{tab:mVio}, although all pairs of $t[Y]$ values share a common class, (i.e., \{\emph{v,w}\}: \E{D}, \{\emph{v,z}\}: \E{C}, \{\emph{w,z}\}: \E{F}), the intersection of these three classes is empty. Hence, $ \{ t_{1}, t_{2} \} \models  \set{X} \rightarrow \set{Y}$, $ \{ t_{2}, t_{3} \} \models  \set{X} \rightarrow \set{Y}$ and $ \{ t_{1}, t_{3} \} \models \set{X} \rightarrow \set{Y}$, but $ \{ t_{1}, t_{2}, t_{3} \} \not \models \set{X} \rightarrow \set{Y}$. However, their logical inference is equivalent. Despite different bases for OFD and NFD axioms, one can show that the axiom systems are equivalent.

\begin{theorem}\label{theorem:equiv}
The axiom systems for OFDs in Theorem~\ref{theorem:complete} and NFDs in Theorem~\ref{theorem:NFDcomplete} are equivalent.
\end{theorem}

\begin{proof}
To prove equivalency, we show that all OFD axioms can be proven from NFD axioms and vice versa. 
\begin{enumerate}
    \item O1.Identity: follows from N1.Reflexivity.
    \item O2.Decomposition: can be inferred from N4.Simplification. 
    \item O3.Composition: since  $X$ $\rightarrow$ $Y$ and $Z$ $\rightarrow$ $W$, it follows from N2.Append that $XZ$ $\rightarrow$ $Y$ and $XZ$ $\rightarrow$ $W$. Hence, by N3.Union, $XZ$ $\rightarrow$ $YZ$ is true.
\end{enumerate}
The other direction: 
\begin{enumerate}
    \item N1.Reflexivity: $X$ $\rightarrow$ $X$ follows by O1.Identity, thus, by O2.Decomposition it can be inferred that $X$ $\rightarrow$ $Y$.
    \item N2.Append: since $Z$ $\subseteq$ $W$, by O1.Identity and O2.Decomposition, it follows that $W$ $\rightarrow$ $Z$. Thus, by O3.Composition $XW$ $\rightarrow$ $YZ$.
    \item N3.Union: since $X$ $\rightarrow$ $Y$ and $X$ $\rightarrow$ $Z$, it follows by O3.Composition that $X$ $\rightarrow$ $YZ$.
    \item N4.Simplification: since $Y$ $\subseteq$ $YZ$, $Z$ $\subseteq$ $YZ$, by O2.Decomposition, we have $X$ $\rightarrow$ $Y$,  $X$ $\rightarrow$ $Z$.
\end{enumerate}
\end{proof}

Theorem~\ref{theorem:equiv} enables us to apply existing algorithms for NFDs to determine whether an OFD holds. Similar conclusions were reached for other classes of dependencies, such as FDs and \emph{pointwise order functional dependencies} (POFDs). While FDs and POFDs are semantically different, the introduced PODs axioms are equivalent to FD axioms,  leading to the same inference~\cite{Wi2001}.


Algorithm~\ref{alg:inference} computes the closure of a set of attributes given a set of OFDs. The inference procedure for NFDs, due to equivalency of inference systems, can be applied to discovered and subsequently user-refined OFDs to ensure continued minimality.  

\begin{theorem}\rev{\emph{~\cite{Lien82}}}
\label{thm:inference}
Algorithm~\ref{alg:inference} computes, in linear time, the closure $X^{+}$, $X^{+}$ $=$ $\{\A{A}$ $|$ $\Sigma$ $\vdash$ $X$ $\rightarrow$ $\A{A} \}$ , where $\Sigma$ denotes a set of OFDs.
\end{theorem}

For a given set of OFDs $\Sigma$, we can find an equivalent \emph{minimal} set, as defined below.

\begin{definition}
\label{def:cover}%
A set $\Sigma$ of OFDs is minimal if
\begin{enumerate}
\item
    $\forall$ $\set{X} \rightarrow \set{Y} \in \Sigma$, $\set{Y}$ is a single attribute; \label{cond:one}
\item
     For no $\set{X}$ $\rightarrow$ $\A{A}$ and a proper subset $\set{Z}$ of $\set{X}$ is
    $\Sigma$ $\setminus$ $\{\set{X}$ $\rightarrow$ $\A{A}\}$ $\cup$ $\{\set{Z}$ $\rightarrow$ $\A{A}\}$
    equivalent to $\Sigma$; \label{cond:three}
\item
    For no $\set{X}$ $\rightarrow$ $\set{Y} \in \Sigma$ is $\Sigma$
    $\setminus$ $\{\set{X}$ $\rightarrow$ $\A{A}\}$ equivalent to $\Sigma$. \label{cond:two}
    \end{enumerate}
If $\Sigma$ is minimal and equivalent to a
set of OFDs $\Sigma'$, then we say $\Sigma$ is a minimal cover of $\Sigma'$.
\end{definition}


\eatTR{\begin{proof}
By the Union and Decomposition inference rules, it is possible to have $\Sigma$ with only a single attribute in the right hand side. We can achieve conditions two other conditions by repeatedly
deleting an attribute and then repeatedly removing a dependency. We can test whether
an attribute $\A{B}$ from $\set{X}$ is redundant for the OFD $\set{X}$ $\rightarrow$ $\A{A}$
by checking if $\A{A}$ is in $\{\set{X} \setminus \A{B}\}^{+}$. We can test whether $\set{X}$ $\rightarrow$ $\A{A}$ is redundant by computing closure $\set{X}^{+}$ with respect to $\Sigma$ $\setminus$ $\{\set{X}$
$\rightarrow$ $\A{A}\}$. Therefore, we eventually reach a set of OFDs which is
equivalent to $\Sigma$ and satisfies conditions \ref{cond:one},
\ref{cond:three} and \ref{cond:two}.
\end{proof}}

\begin{example}
Let $\Sigma$ = $\{ \Sigma_{1}: \A{CC}$ $\rightarrow$ $\A{CTRY}$, $\{ \Sigma_{2}: \A{CC}$, $\A{DIAG} \}$ $\rightarrow$ $\A{MED}$, $\{ \Sigma_{3}:$ $\A{CC}$, $\A{DIAG} \}$ $\rightarrow$ $\{ \A{MED}$, $\A{CTRY}\} \}$. This set is not a minimal cover as $\Sigma_{3}$ follows from $\Sigma_{1}$ and $\Sigma_{2}$ by Composition.
\end{example}

\subsection{Optimizations}
\label{sec:3_opt}

As we will show in Section~\ref{sec:fastofd}, the search space of potential OFDs is exponential in the number of attributes, as with traditional FDs.  To improve the efficiency of OFD discovery, we use axioms to prune redundant and non-minimal OFDs.


\begin{lemma} (Opt-1) If $A$ $\in$ $X$ then $X \rightarrow A$.
\vspace{-0.2cm}
\end{lemma}

\eatTR{\begin{proof}
    It follows from Reflexivity (Lemma~\ref{lemma:reflexivity})
\end{proof}}
If $A \in X$ then $X \rightarrow A$ is a trivial dependency (Reflexivity).

\begin{lemma} (Opt-2) If  $X \rightarrow A$ is satisfied over $I$, then $XY \rightarrow A$ is satisfied for all $Y \subseteq R \setminus X$.
\vspace{-0.2cm}
\end{lemma} 
\eatTR{\begin{proof}
    Assume $X \rightarrow A$. The OFD $X \rightarrow \{ \}$ follows from Reflexivity ((Lemma~\ref{lemma:reflexivity})). Hence, it can be inferred by Composition that $XY \rightarrow A$.
\end{proof}}
If $X \rightarrow A$ holds in $I$, then all OFDs containing supersets of $X$  also hold in $I$ (Augmentation), and can be pruned.  When we identify a key during OFD search, we can apply additional optimizations.

\begin{lemma} (Opt-3) If $X$ is a key (or superkey) in $I$, then for any $A$,  $X \rightarrow A$ is satisfied in $I$.
\vspace{-0.2cm}
\end{lemma}
\eatTR{\begin{proof}
    Since $X$ is a super-key, partition $\Pi_{X}$ consists of singleton equivalence classes only. Hence, the OFD $X \rightarrow A$ is valid.
\end{proof}}
For a candidate OFD $\set{X} \rightarrow \A{A}$, if $X$ is a key, then for all $x \in \Pi_{X}$, $|x|$ = 1, and $X \rightarrow A$ always holds.  On the other hand, if $\set{X} \setminus \A{A}$ is a superkey but not a key, then clearly the OFD $\set{X} \rightarrow \A{A}$ is not minimal. This is because there exists $\A{B} \in \set{X}$, such that $\set{X} \setminus \A{B}$ is a superkey and $\set{X} \setminus \A{B} \rightarrow  \A{A}$ holds.

\begin{lemma} (Opt-4)  
If all tuples in an equivalence class $x \in \Pi_{X}$ have the same value of $A$, then a traditional FD, and therefore an OFD, is satisfied in $x$.
\vspace{-0.2cm}
\end{lemma} 

\eatTR{\begin{proof}
Singleton equivalence classes over attribute set $X$ cannot falsify any OFD $X \rightarrow A$. \\
\end{proof}
}

A \emph{stripped partition} of $\Pi_{\set{X}}$, denoted   $\Pi^{*}_{\set{X}}$, removes all the equivalence classes of size one.  
For example,
in Table~\ref{tab:mVio}, $\Pi_{\A{CC}}$ $=$ $\{ \brac{t_{1}, t_{5},t_{6}}$, $\brac{t_{2}, t_{4},t_{7}},$ $\brac{t_{3}} \}$, whereas the stripped partition removes the singleton equivalence class \{$t_{3}$\}, so $\Pi^{*}_{\A{CC}}$ $=$ $\{ \brac{t_{1}, t_{5},t_{6}}, \brac{t_{2}, t_{4},t_{7}} \}$.  
If $\Pi^{*}_{\set{X}}$ = $\emptySet$, then $\set{X}$ is a superkey and Optimization 3 applies.

\begin{lemma}\label{lemma:stripped}
Singleton equivalence classes over attribute set $\set{X}$ cannot violate any OFD $\set{X} \rightarrow \A{A}$.
\end{lemma}

\eatTR{
\begin{proof}
Follows directly from the definition of OFDs.  
\end{proof}
}

\section{OFD Discovery}
\label{sec:fastofd}

We now present an algorithm to discover a complete and minimal set of OFDs from data. Based on our axiomatization for OFDs, 
we normalize all OFDs to a single attribute consequent, i.e., $X$ $\rightarrow$ $A$ for any attribute $A$.  An OFD $\set{X} \rightarrow \A{A}$ is \emph{trivial} if $\A{A} \in \set{X}$ by Reflexivity. An OFD $\set{X} \rightarrow \A{A}$ is \emph{minimal} if it is non-trivial and there is no set of attributes $\set{Y} \subset \set{X}$ such that $\set{Y} \rightarrow \A{A}$ holds by Augmentation.


\begin{figure}[!t]
\noindent\begin{minipage}[t]{.5\textwidth}
\begin{algorithm}[H]
\centering
\caption{Inference procedure for OFDs}\label{alg:inference}
\raggedright{Input}: A set of OFDs $\Sigma$, and a set of attributes
    $\set{X}$.\\
{Output}: The closure of $\set{X}$ with respect to $\Sigma$. \\
  \footnotesize
\begin{algorithmic}[1]
    \STATE $\Sigma_{unused}$ $\leftarrow$ $\Sigma$%
    \STATE $n$ $\leftarrow$ $0$
    \STATE $\set{X}^{n}$ $\leftarrow$ $\set{X}$
    \LOOP
    \IF{$\exists$ $\set{V}$ $\rightarrow$ $\set{Z}$ $\in$ \label{lines:st}
                    $\Sigma_{unused}$ and $\set{V}$ $\subseteq$ $\set{X}$}
            \STATE $\set{X}^{n+1}$ $\leftarrow$ $\set{X}^{n}$ $\cup$ $\set{Z}$
            \STATE $\Sigma_{unused}$ $\leftarrow$
            $\Sigma_{unused}$ $\setminus$ \{$\set{V}$ $\rightarrow$ $\set{Z}$\}
            \STATE $n$ $\leftarrow$ $n+1$ \label{lines:end}
        \ELSE
            \RETURN \emph{X}$^{n}$
        \ENDIF
    \ENDLOOP \STATE \textbf{end loop}
\end{algorithmic}
\end{algorithm}
\end{minipage}%
\hfill
\begin{minipage}[t]{0.46\textwidth}
\begin{algorithm}[H]
    \centering
     \caption{\fastofd} \label{pc:OFDdiscov}
     \raggedright{Input:} Relation $\R{r}$ over schema $\R{R}$\\
   {Output:} Minimal set of OFDs $\Sigma$, s.t. $\R{r}$ $\models$ $\Sigma$ \\
    \footnotesize
   \begin{algorithmic}[1]
        \STATE $\set{L}_{0}$ $=$ $\emptySet{}$, $\Sigma$ $=$ $\emptySet{}$
        \STATE $\set{C}^{+}(\emptySet{}) = \R{R}$
        \STATE $l$ $=$ $1$
        \STATE $\set{L}_{1}$ $=$ $\{ \A{A}$ $|$ $\A{A} \in \R{R} \}$
        \WHILE{$\set{L}_{l} \not= \emptySet$} \label{line:mainLoop}
            \STATE \emph{computeOFDs}($\set{L}_{l}$)
            \STATE $\set{L}_{l+1}$ $=$ \emph{calculateNextLevel}($\set{L}_{l}$)
            \STATE $l$ $=$ $l+1$
        \ENDWHILE \label{line:endMainLoop}
        \RETURN $\Sigma$
    \end{algorithmic}
\end{algorithm}
\end{minipage}
\vspace{-0.3cm}
\end{figure}

The set of possible antecedent (left side) and consequent (right side) values considered by our algorithm can be modeled as a set containment lattice.
\eat{For example, Figure \ref{fig:lattice} shows the search lattice for four of the five attributes in Table \ref{tab:example}.}  Each node in the lattice represents an attribute set and an edge exists between sets $X$ and $Y$ if $X$ $\subset$ $Y$ and $Y$ has exactly one more attribute than $X$.  
Let $k$ be the number of levels in the lattice.  A relation with $n$ attributes generates a $k = n$ level lattice, with $k = 0$ representing the top (root node) level.  

\eat{ 
\begin{figure}[t]
\centering
       \includegraphics[width=2.6in]{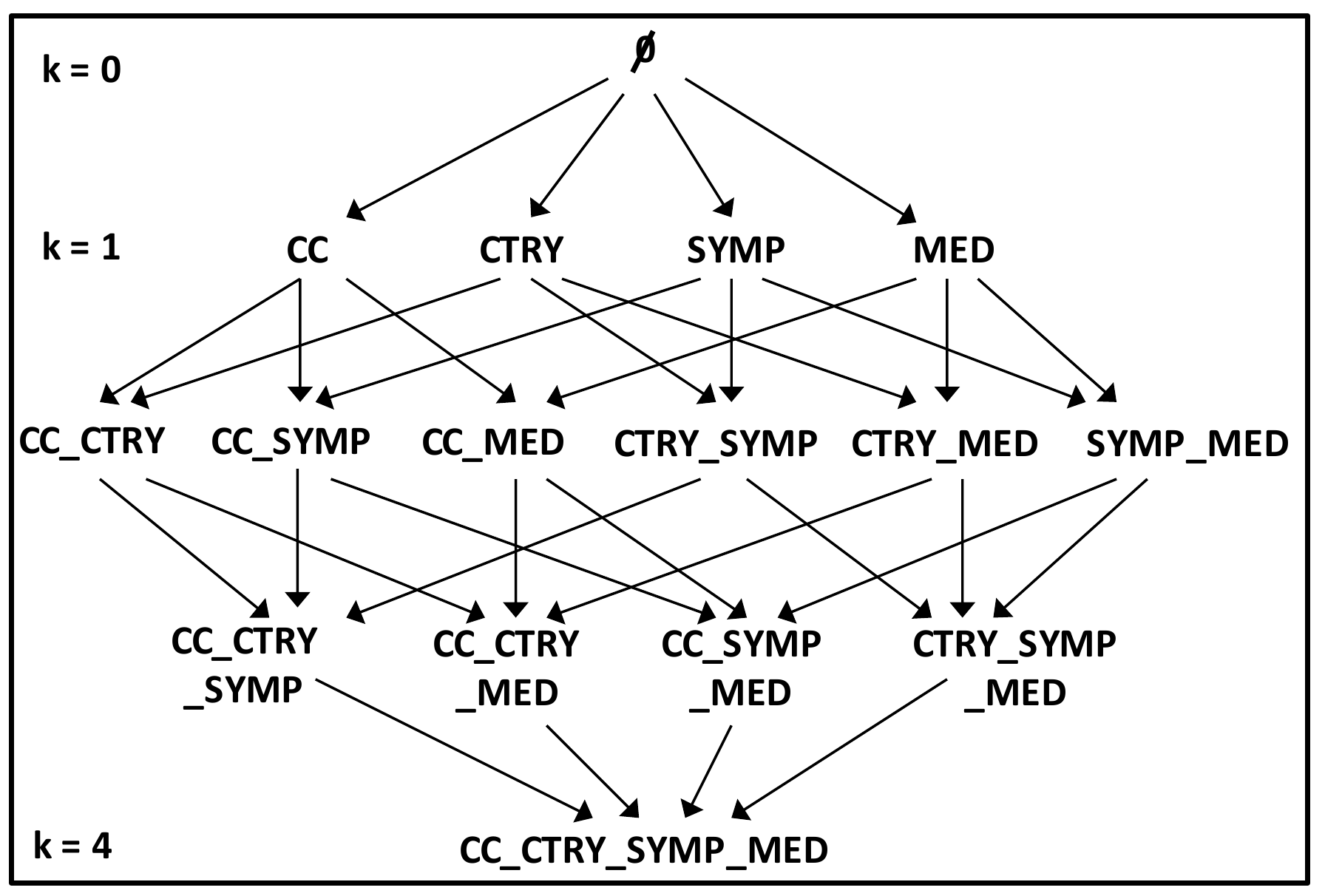}
     \vspace{-0.2cm}
    \caption{Example of an attribute lattice. \label{fig:lattice}}
\end{figure}
}

After computing the stripped partitions $\Pi^*_{A}, \Pi^*_{B} \dots$ for single attributes at level $k = 1$, we compute the stripped partitions for subsequent levels in linear time by taking the product, i.e., $\Pi^*_{AB}$ = $\Pi^*_{A} \cdot \Pi^*_{B}$.  OFD candidates are considered by traversing the lattice in a breadth-first search manner.  We consider all $X$ consisting of single attribute sets, followed by all 2-attribute sets, and continue level by level until (potentially) level $k$ = $n$, similarly as for other use cases of Apriori~\cite{AMW96}.  

Algorithm~\ref{pc:OFDdiscov} outlines the OFD discovery process.  In level  $\set{L}_{l}$ of the lattice, we generate candidate OFDs with $l$ attributes using  \emph{computeOFDs}($\set{L}_{l})$. \fastofd starts from singleton sets of attributes and works its way to larger attribute sets through the lattice, level by level. When the algorithm processes an attribute set $\set{X}$, it verifies candidate OFDs of the form ($\set{X} \setminus \A{A}) \rightarrow \A{A}$, where $\A{A} \in \set{X}$.  This guarantees that only non-trivial OFDs are considered.  For each candidate, we check if it is a valid synonym OFD.
The small-to-large search strategy guarantees that only minimal OFDs are added to the output set $\Sigma$, and is used to prune the search space. The OFD candidates generated in a given level are checked for \emph{minimality} based on the previous levels and are added to a \emph{valid} set of OFDs $\Sigma$ if applicable. The algorithm \emph{calculateNextLevel}($\set{L}_{l}$) forms the next level from the current level. Next, we explain how we check for minimality, and routines \emph{calculateNextLevel} and \emph{computeOFDs}.

\subsection{Finding Minimal OFDs}\label{sec:optimizations}

\fastofd traverses the lattice until all minimal OFDs are found. \eat{We deal with OFDs of the form $\set{X} \setminus \A{A} \rightarrow \A{A}$, where $\A{A} \in \set{X}$.}  To check if an OFD is minimal, we need to know if $\set{X} \setminus \A{A} \rightarrow \A{A}$ is valid for any $\set{Y} \subset \set{X}$. If $\set{Y} \setminus \A{A} \rightarrow \A{A}$, then by Augmentation $\set{X} \setminus \A{A} \rightarrow \A{A}$ holds. An OFD $\set{X} \rightarrow \A{A}$ holds for any relational instance by Reflexivity, therefore, considering only $\set{X} \setminus \A{A} \rightarrow \A{A}$ guarantees that only non-trivial OFDs are taken into account.

We maintain information about minimal OFDs, in the form of $\set{X} \setminus \A{A} \rightarrow \A{A}$, in the \emph{candidate} set $\set{C}^{+}(\set{X})$. If $\A{A} \in \set{C}^{+}(\set{X})$ for a given set $\set{X}$, then $\A{A}$ has not been found to depend on any proper subset of $\set{X}$. Therefore, to find minimal OFDs, it suffices to verify OFDs $\set{X} \setminus \A{A} \rightarrow \A{A}$, where $\A{A} \in \set{X}$ and $\A{A} \in \set{C}^{+}(\set{X} \setminus \A{B})$ for all $\A{B} \in \set{X}$.

\begin{example}
Assume that $\A{B} \rightarrow \A{A}$ and that we consider the set $\set{X}$ = $\brac{\A{A}, \A{B}, \A{C}}$. As $\A{B} \rightarrow \A{A}$ holds, $\A{A} \not \in$ $\set{C}^{+}(\set{X} \setminus \A{C})$. Hence, the OFD $\brac{\A{B}, \A{C}} \rightarrow \A{A}$ is not minimal.
\end{example}

\begin{definition}\label{def:minConst}
$\set{C}^{+}(\set{X})$ $=$ $\{ \A{A} \in \R{R}$ $|$ $\forall_{\A{A} \in \set{X}}$  $\set{X} \setminus \A{A} \rightarrow \A{A}$ does not hold$\}$.
\end{definition}

\reviseTwo{Some of our techniques are similar to \emph{TANE}~\cite{HKPT98} for FD discovery and \emph{FASTOD}~\cite{ODTR16} for Order Dependency (OD) discovery since OFDs subsume FDs and ODs subsume FDs. However, \fastofd differs from \emph{TANE} and \emph{FASTOD} in the  optimizations, how nodes are removed from the lattice, and applying keys to prune candidates.  \emph{\fastofd} includes OFD-specific rules. For instance, for FDs if $\brac{\A{B}, \A{C}} \rightarrow \A{A}$ and $\A{B} \rightarrow \A{C}$, then $\A{B} \rightarrow \A{A}$ holds; hence, $\brac{\A{B}, \A{C}} \rightarrow \A{A}$ is non-minimal. However, this rule does not hold for OFDs, and therefore our definition of candidate set $\set{C}^{+}(\set{X})$ differs from \emph{TANE}.}

\begin{figure}[!t]
\noindent\begin{minipage}[t]{.5\textwidth}
\setlength{\textfloatsep}{0.1cm}
\begin{algorithm}[H]
\centering
  \caption{calculateNextLevel($\mathcal{L}_{l}$)} \label{pc:nextLevel}
  \footnotesize
 \begin{algorithmic}[1]
        \STATE $\set{L}_{l+1} = \emptySet{}$
        \FORALL{$\brac{\set{Y}\A{B}, \set{Y}\A{C}} \in \emph{singleAttrDiffBlocks}(\set{L}_{l})$} \label{line:blocks}
            \STATE $\set{X} = \set{Y} \bigcup \brac{\A{B}, \A{C}}$
            \label{line:prevLevel}
            \STATE Add $\set{X}$ to $\set{L}_{l+1}$
        \ENDFOR
        \RETURN $\set{L}_{l+1}$
    \end{algorithmic}
\end{algorithm}
\end{minipage}%
\hfill
\begin{minipage}[t]{0.46\textwidth}
\begin{algorithm}[H]
    \centering
    \caption{computeOFDs($\mathcal{L}_{l}$)} \label{pc:computeODs}
    \footnotesize
   \begin{algorithmic}[1]
        \FORALL{$\set{X} \in \set{L}_{l}$} \label{line:loopLevel}
                \STATE $\set{C}^{+}(\set{X}) = \bigcap_{\A{A} \in \set{X}} \set{C}^{+}(\set{X} \setminus \A{A})$ \label{line:interConst}
        \ENDFOR

        \FORALL{$\set{X} \in \set{L}_{l}$}
            \FORALL{$\A{A} \in \set{X} \cap \set{C}^{+}(\set{X})$}
                    \label{line:candConst}
                \IF{$\set{X} \setminus \A{A} \rightarrow \A{A}$}
                        \label{line:checkConst}
                    \STATE Add $\set{X} \setminus \A{A} \rightarrow \A{A}$
                            to $\Sigma$ \label{line:addConst}
                    \STATE Remove $\A{A}$ from $\set{C}^{+}(\set{X})$
                    \label{line:removeAttr}
                \ENDIF
            \ENDFOR
        \ENDFOR
    \end{algorithmic}
\end{algorithm}
\setlength{\floatsep}{0.1cm}
\end{minipage}
\vspace{-0.35cm}
\end{figure}

\subsection{Computing Levels}\label{sec:calcLev}

Algorithm~\ref{pc:nextLevel} explains \emph{calculateNextLevel}($\set{L}_{l}$), which computes $\set{L}_{l+1}$ from $\set{L}_{l}$. It uses the subroutine \emph{singleAttrDifferBlocks}($\set{L}_{l}$) that partitions $\set{L}_{l}$ into blocks (Line~\ref{line:blocks}). Two sets belong to the same block if they have a common subset $\set{Y}$ of length $l-1$ and differ in only one attribute, $\A{A}$ and $\A{B}$, respectively. Therefore, the blocks are not difficult to calculate as sets $\set{Y}\A{A}$ and $\set{Y}\A{B}$ can be expressed as sorted sets of attributes. Other usual use cases of Apriori~\cite{AMW96} such as \emph{TANE}~\cite{HKPT98} 
use a similar approach.

The level  $\set{L}_{l+1}$ contains those sets of attributes of size $l + 1$ which have their subsets of size $l$ in $\set{L}_{l}$.

\subsection{Computing Dependencies \& Completeness}
\label{sec:dependencies}

Algorithm~\ref{pc:computeODs} adds minimal OFDs from level $\set{L}_l$ to $\Sigma$, in the form of $\set{X} \setminus \A{A} \rightarrow \A{A}$, where $\A{A} \in \set{X}$. The following lemma shows that we can use candidates $\set{C}^{+}(\set{X})$ to test whether  $\set{X} \setminus \A{A} \rightarrow \A{A}$ is minimal.

\begin{lemma}\label{minConst}
An OFD $\set{X} \setminus \A{A} \rightarrow \A{A}$, where $\A{A} \in \set{X}$, is minimal iff $\forall_{\A{B} \in \set{X}} \A{A} \in \set{C}^{+}(\set{X} \setminus \A{B})$.
\end{lemma}

\eatTR{
\begin{proof}
Assume first that the dependency $\set{X} \setminus \A{A} \rightarrow \A{A}$ is not minimal. Therefore, there exists $\A{B} \in \set{X}$ for which $\set{X} \setminus \brac{\A{A}, \A{B}} \rightarrow \A{A}$ holds. Then, $\A{A} \not \in \set{C}^{+}(\set{X} \setminus \A{B})$.

To prove the other direction assume that there exists $\A{B} \in \set{X}$, such that $\A{A} \not \in \set{C}^{+}(\set{X} \setminus \A{B})$. Therefore, $\set{X} \setminus \brac{\A{A}, \A{B}} \rightarrow \A{A}$ holds, where $\A{A} \not = \A{B}$. Hence, by Reflexivity the dependency $\set{X} \setminus \A{A} \rightarrow \A{A}$ is not minimal.
\end{proof}
}

By Lemma~\ref{minConst}, the steps in Lines~\ref{line:interConst}, \ref{line:candConst}, \ref{line:checkConst} and \ref{line:addConst} guarantee that the algorithm adds to $\Sigma$ only the minimal OFDs of the form $\set{X} \setminus \A{A} \rightarrow \A{A}$, where $\set{X} \in \set{L}_{l}$ and $\A{A} \in \set{X}$.   In Line \ref{line:checkConst}, to verify whether  $\set{X} \setminus \A{A} \rightarrow \A{A}$ is a synonym OFD, we apply Definition~\ref{def:synonym_ofd}.



The worst case complexity of our algorithm is \emph{exponential} in the number of attributes as there are $2^{|n|}$ nodes in the lattice.  The worst-case output size is also exponential in the number of attributes, and occurs when the minimal OFDs are in the widest middle level of the lattice.  This means that a polynomial-time discovery algorithm in the number of attributes cannot exist.  These results are in line with previous FD \cite{HKPT98}, inclusion dependency \cite{PKQJN15}, and order dependency~\cite{ODTR16} discovery algorithms.  However, the complexity is polynomial in the number of tuples, although the ontological relationships (synonyms) influence the complexity of verifying whether a candidate OFD holds. We assume that values in the ontology are indexed and can be accessed in constant time.

Ontology FD candidate verification differs from traditional FDs. Following Definition~\ref{def:synonym_ofd} to verify whether a candidate synonym OFD holds over $I$, for each equivalence class $x \in \Pi_{X}(I)$, we need to check whether the intersection of the corresponding senses is not empty. This can be done in linear time in the number of tuples by scanning the stripped partitions and maintaining a hash table with the frequency counts of all the senses for each equivalence class.  Returning to the example in Table~2, the synonym OFD $X \rightarrow_{syn} Y$ does not hold because for the single equivalence class in this example, of size three, there are no senses (classes) for $Y$ that appear three times.  

\begin{lemma}\label{lemma:levelCorrSplit}
Let candidates $\set{C}^{+}(\set{Y})$ be correctly computed $\forall \set{Y} \in \set{L}_{l-1}$.  $computeOFDs$($\set{L}_{l}$) calculates correctly $\set{C}^{+}(\set{X})$, $\forall \set{X} \in \set{L}_{l}$.
\end{lemma}

\eatTR{
\begin{proof}
An attribute $\A{A}$ is in $\set{C}^{+}(\set{X})$ after the execution of the algorithm $computeOFDs$($\set{L}_{l})$  unless it is excluded from $\set{C}^{+}(\set{X})$ on Line \ref{line:interConst} or \ref{line:removeAttr}. First we show that if $\A{A}$ is excluded from $\set{C}^{+}(\set{X})$ by $computeOFDs$($\set{L}_{l})$, then $\A{A} \not \in$ $\set{C}^{+}(\set{X})$ by the definition of $\set{C}^{+}(\set{X})$.
    \begin{itemize}[nolistsep]
        \item[-] If $\A{A}$ is excluded from $\set{C}^{+}(\set{X})$ on Line~\ref{line:interConst}, there exists $\A{B} \in \set{X}$ with $\A{A} \not \in \set{C}^{+}(\set{X} \setminus \A{B})$. Therefore, $\set{X} \setminus \brac{\A{A}, \A{B}} \rightarrow \A{A}$ holds, where $\A{A} \not = \A{B}$. Hence, $\A{A} \not \in$ $\set{C}^{+}(\set{X})$ by the definition of $\set{C}^{+}(\set{X})$. 
        \item[-] If $\A{A}$ is excluded on Line~\ref{line:removeAttr}, then $\A{A} \in \set{X}$ and $\set{X} \setminus \A{A} \rightarrow \A{A}$ holds. Hence, $\A{A} \not \in$ $\set{C}^{+}(\set{X})$ by the definition of $\set{C}^{+}(\set{X})$.
    \end{itemize}
Next, we show the other direction, that if $\A{A} \not \in$ $\set{C}^{+}(\set{X})$ by the definition of $\set{C}^{+}(\set{X})$, then $\A{A}$ is excluded from $\set{C}^{+}(\set{X})$ by the algorithm $computeOFDs$($\set{L}_{l}$). Assume $\A{A} \not \in$ $\set{C}^{+}(\set{X})$ by the definition of $\set{C}^{+}(\set{X})$. Therefore, there exists $\A{B} \in \set{X}$, such that $\set{X} \setminus \brac{\A{A}, \A{B}} \rightarrow \A{A}$ holds. We have following two cases.
    \begin{itemize}[nolistsep]
        \item[-] $\A{A} = \A{B}$. Thus, $\set{X} \setminus \A{A} \rightarrow \A{A}$ holds and $\A{A}$ is removed on Line~\ref{line:removeAttr}, if $\set{X} \setminus \A{A} \rightarrow \A{A}$ is minimal; and on Line~\ref{line:interConst} otherwise. 
        \item[-] $\A{A} \not = \A{B}$. Hence, $\A{A} \not \in$ $\set{C}^{+}(\set{X} \setminus \A{B})$ and $\A{A}$ is removed on Line~\ref{line:interConst}.
    \end{itemize}
This ends the proof of correctness of computing the candidate set $\set{C}^{+}(\set{X})$, $\forall \set{X} \in \set{L}_{l}$.
\end{proof}
}


\begin{theorem}\label{theorem:completeness}
The \fastofd algorithm computes a complete and minimal set of OFDs $\Sigma$.
\end{theorem}

\eatTR{
\begin{proof}
The algorithm $computeOFDs$($\set{L}_{l})$ adds to set of OFDs $\Sigma$ only the minimal OFDs. The steps in Lines~\ref{line:interConst}, \ref{line:candConst}, \ref{line:checkConst} and \ref{line:addConst} guarantee that the algorithm adds to $\Sigma$ only the minimal OFDs of the form $\set{X} \setminus \A{A} \rightarrow \A{A}$, where $\set{X} \in \set{L}_{l}$ and $\A{A} \in \set{X}$ by Lemma~\ref{minConst}. It follows by induction that $computeOFDs$($\set{L}_{l})$ calculates correctly $\set{C}^{+}(\set{X})$ for all levels $l$ of the lattice since Lemma~\ref{lemma:levelCorrSplit} holds. Therefore, the \fastofd algorithm computes a complete set of minimal OFDs $\Sigma$.
\end{proof}
}

\section{Data Cleaning with OFDs} \label{sec:cleaning}

Over time, misalignment may arise between a data instance $I$, an ontology $S$, and a set of OFDs $\Sigma$.  In this section, we study the problem of how to compute repairs to re-align $I, S$ with $\Sigma$.  Data naturally evolve due to updates and changes in domain semantics.  These changes in semantics also lead to ontology incompleteness, as new concepts and ontological relationships are introduced.  For example, new uses of medical drugs lead to new prescriptions to treat new and different illnesses.  While OFDs may also evolve, we argue that repairs to the data $I$ and ontology $S$ are more common in practice, and hence, our current focus.  When stale OFDs do occur, \fastofd can be used to discover the latest set of OFDs that hold over the data.  \reviseTwo{For approximate OFDs defined over a dirty instance $I$, violating values in $I$ can be repaired, thereby transforming approximate OFDs to OFDs that are satisfied over all tuples in $I$.}  Formally, we modify $I$ and $S$ to produce $I'$ and $S'$, such that a repaired version of $I$, $I' \models \Sigma$, w.r.t. a modified version $S'$ of $S$.  

\subsection{Scope of Repairs}
\label{sec:repairscope}


OFDs interact when two or more dependencies share a common set of attributes, which may invalidate already performed repairs.  For simplicity, we assume that there does not exist an attribute that occurs on the left side of one OFD and the right side of another.  \reviseTwo{This simplifies the interaction among repairs to only consequent attributes without having to consider changes to antecedent (left-side) attributes of the dependency, i.e., the equivalence classes for an OFD remain fixed.}  \blue{In many domains, attributes serve as either an independent or dependent role.  For example, in health care, demographic and prognosis characteristics are typically independent and dependent attributes, respectively.}  However, our techniques can handle OFDs that share the same consequent attribute.  

\noindent \textbf{Data Repair.}  Consider a synonym OFD $\phi$: $X \rightarrow_{syn} A$, and an equivalence class $x \in \Pi_{X}(I)$ where there exists at least two tuples $t_1, t_2 \in x$, such that $t_1[A]$ and $t_2[A]$ are not synonyms under some sense $\lambda$. \eat{We consider updates to $t_1[A]$ where the domain of repair values are values in $S$ under sense $\lambda$ to ensure all $t[A]$ values in $x$ are synonyms.}  We consider repairs to values in the consequent attributes where the domain of repair values are values in $S$. Let $\mathcal{P}(I)$ denote the set of all possible data repairs of $I$. For a repair $I'$ of $I$, we use $dist(I, I')$ to denote the difference between $I$ and $I'$, which is measured by the number of cells whose values in $I$ are different from  the values of the corresponding cells in $I'$.

%
\noindent\textbf{Ontology Repair.} 
\reviseTwo{A repair to an ontology $S$ is the insertion of new value(s) to a node in $S$ w.r.t. a sense $\lambda$.}  Let $\mathcal{P}(S)$ denote the set of possible ontology repairs to $S$.  For a repair $S'$ of $S$, we define distance $dist(S, S')$ 
as the number of new values (\emph{concepts}) in $S'$ that are not in $S$.  

\noindent \textbf{Minimal Repair. }
The universe of repairs, $U$, represents all possible pairs ($S', I'$) such that $S' \in \mathcal{P}(S)$ and $I' \in \mathcal{P}(I)$ and $I' \models \Sigma$ w.r.t. $S'$.  
We want to find minimal repairs in $U$ that do not make unnecessary changes to $S$ or $I$ in a Pareto-optimal sense.

\begin{definition}\label{def:minimal} (\emph{Minimal Repair}). 
We are given an instance $I$, OFDs $\Sigma$, and an ontology $S$, where $I \not\models \Sigma$ \textit{w.r.t.} $S$.   A repair ($S'$, $I'$) $\in U$ is minimal \eat{under a sense $\lambda$} if $\nexists (S'', I'') \in U$ such that $dist(S,S'') < dist(S,S')$ and $dist(I,I'') < dist(I,I')$.
\end{definition}

For example, suppose $U$ contains three repairs: the first one makes two changes each to $S$ and $I$, the second one makes two changes to $S$ and three changes to $I$, and the third one makes one change to $S$ and five changes to $I$.  Here, the first and the third repair are minimal.

\begin{figure}[!t]
\noindent\begin{minipage}{.5\textwidth}
  \centering
  \rule{0.3\textwidth}{0pt}
    \includegraphics[width=6.9cm]{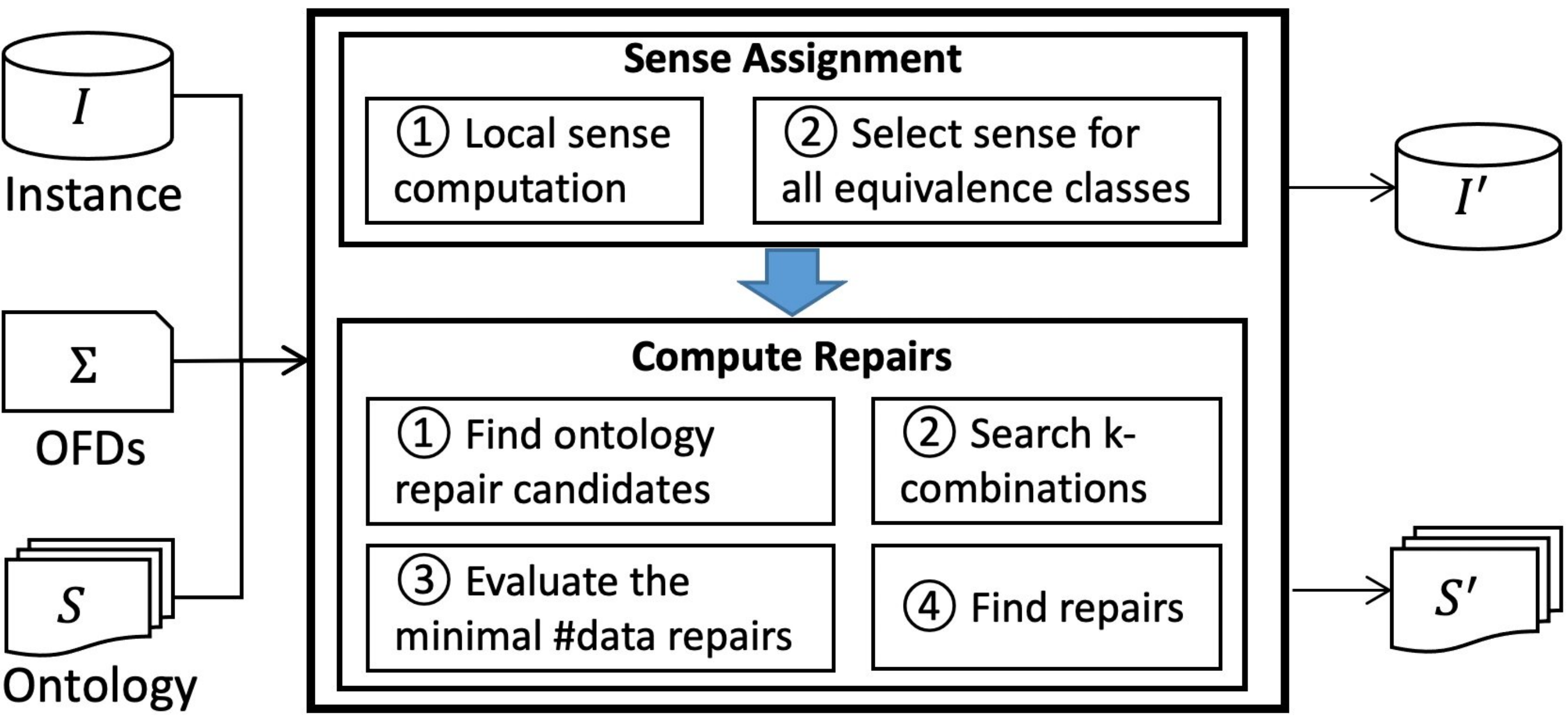}
  \captionof{figure}{OFDClean framework.} \label{fig:repairframework}
\end{minipage}%
\begin{minipage}{.5\textwidth}
  \centering
  \rule{0.3\textwidth}{0pt}
    \includegraphics[width=6.9cm]{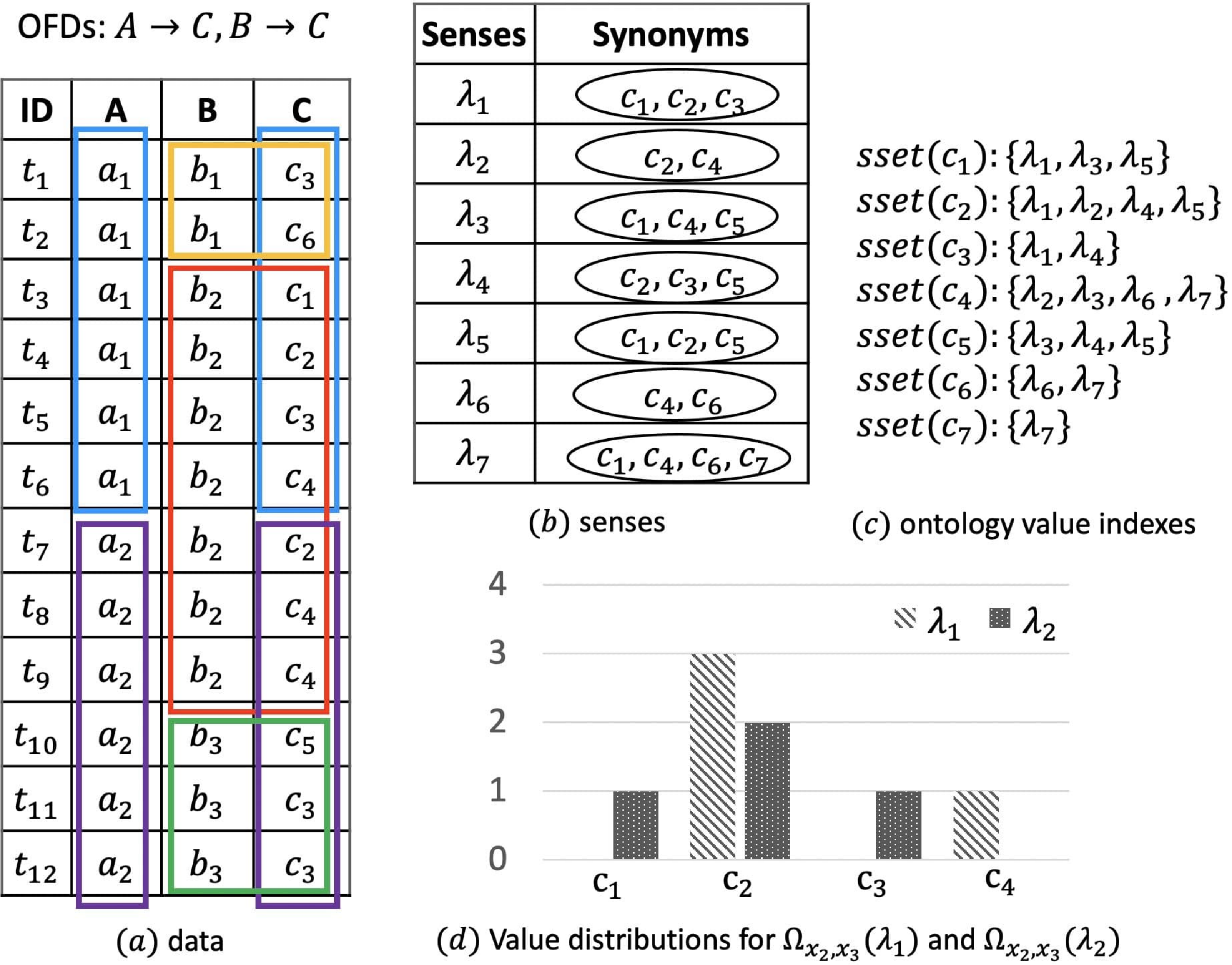}
  \captionof{figure}{Sense assignment.} \label{fig:ecg}
\end{minipage}
\vspace{-0.6cm}
\end{figure}

\reviseThree{
\noindent \textbf{Problem:} \emph{
For a data instance $I$, an ontology $S$, a set of OFDs $\Sigma$, where $I \not \models \Sigma$, compute a repaired instance $I'$, and ontology $S'$, such that $I' \models \Sigma$ w.r.t. $S'$, and $dist(I, I')$ and $dist(S, S')$ are minimal.
}}

\subsection{\ofdclean Overview}
\label{sec:framework}

OFDs enrich the semantic expressiveness of FDs via ontological senses, which allow for multiple interpretations of the data.  In traditional FDs, senses do not exist, and there is only one interpretation.  For a set of OFDs $\Sigma$, the problem of finding a minimal repair to $S$ and to $I$ requires evaluating candidate senses for each equivalence class w.r.t. each $\phi \in \Sigma$, and assigning a sense to each class.   

Existing constraint-based repair algorithms are insufficient as they do not consider senses~\cite{CM11, GMPS13}.  We therefore propose \ofdclean, a framework that includes ontology and data repair with sense assignment to minimize the number of changes to $I$, and $S$.  \reviseTwo{Intuitively, we consider ontology repairs of size $k$, where $k$ ranges from one to the total number of ontology repair values.  For each candidate ontology repair of size $k$, we select the ontology repair $S'$ such that $dist(S, S')$ is minimal.  For each value $k$, and the selected $S'$, we compute a repair $I'$ of $I$ such that the number of updates is minimal, and less than a threshold $\tau$, i.e., $dist(I, I') \leq \tau$ to achieve consistency with $\Sigma$. We call such repairs $\tau$-\emph{constrained repairs}.  To generate a Pareto-optimal set of repairs, we select $(S', I')$ that are minimal for each $k$ value.}  

\ofdclean takes a greedy, two-step approach. First, we compute an initial sense assignment for each equivalence class that minimizes the number of repairs to the ontology and to the data (locally) w.r.t. each OFD and each equivalence class. Secondly, we go beyond a single OFD to consider \reviseTwo{pairwise} interactions among $\phi \in \Sigma$ where attribute overlap occurs, such that repairing values w.r.t. $\phi_1$ may increase or decrease the number of errors in $\phi_2$.  We model the overlapping tuples between $\phi_1$ and $\phi_2$ as two distributions w.r.t. their assigned senses.  We refine the sense assignment such that the number of necessary updates to transform one distribution to the other is minimized.  
Lastly, with a computed sense assignment, our problem reduces to finding a minimal repair of $S$ and to $I$ such that it satisfies $\Sigma$.  Figure~\ref{fig:repairframework} provides an overview of \ofdclean.  Given $I, S$ and $\Sigma$, \ofdclean identifies inconsistencies in $I$ w.r.t. $\Sigma$ and $S$, and returns a set of minimal repairs ($I', S'$) such that $I' \models \Sigma$ w.r.t. $S'$.  \ofdclean execution proceeds along two main modules.


\noindent \textbf{Sense Assignment.} Let $\phi: X \rightarrow_{syn} A$ be a synonym OFD, where for each 
equivalence class $x \in \Pi_{X}$, $x$ is interpreted w.r.t. an assigned sense $\lambda_x$.  \eat{This requires us to select one interpretation during error detection and repair  Finding an optimal data repair, under one interpretation, for a set of FDs is already NP-hard~\cite{BFFR05}, hence, considering all sense assignments for all $x$ is not practically feasible.} Given $I, S$, and $\Sigma$, for each $\phi \in \Sigma$,  we compute $\lambda_x$ for each $x \in \Pi_{X}$ such that $(I', S')$ is minimal and $I' \models \phi$. We denote the set of all senses across all $x$ for $\phi$ as \lambdaphi $ = \cup_{x \in \Pi(X)} {\lambda_x}$.    Intuitively, we compute $\lambda_x$ for every $x \in \Pi_{X}$ to obtain an interpretation that can achieve a minimal repair.  We take a local greedy approach to consider only $\phi$ when computing \lambdaphi.  We then model the interactions among the OFDs in $\Sigma$, via their overlapping equivalence classes, and refine the initial sense assignment.  We quantify interactions and conflicts by comparing the difference between the distributions represented via a pair of equivalence classes.  \eat{For each $x \in \Pi_{X}$, for every $\phi \in \Sigma$, we assign a sense $\lambda$ to x, where $\lambda_x \in \Lambda_{\phi}$. }  The final assignment of senses for all $\phi \in \Sigma$ is denoted as $\Lambda(\Sigma) = \cup_{\phi}$ \lambdaphi.  \reviseTwo{We present our sense assignment algorithm in Section~\ref{sec:sense}.}

\noindent \textbf{Compute Repairs.}   \eat{An ontology $S$ may become outdated as new domain concepts are introduced (but not added to $S$).  For example, new medications are introduced after research and clinical trials, used by patients, but not updated in ontologies.  In such cases, we seek to repair $S$ by adding these concepts to $S$ to re-align with the data.}  
Given $\Lambda(\Sigma)$,  {\sf Compute Repairs} computes the set of minimal repairs by first generating the set of ontology repair candidates that appear in an equivalence class $x$, but not in $S$.  It then considers combinations of these candidates. 
We optimize lattice traversal using a beam-search strategy that selects the top-$b$ candidates to expand at each level of the lattice~\cite{lowerre}. We consider data updates to the consequent attribute $A$ from the domain of $A$ and synonyms of $A$ in $S$ under the selected sense.  We then select candidate repairs that make at most $\tau$ updates.  \reviseTwo{We present details of how ontology and data repairs are computed in Section~\ref{sec:ontrep}.}

\section{Sense Assignment}\label{sec:sense}
Given senses for an ontology $S$, we compute $\lambda_x$ for each $x$ and $\phi \in \Sigma$ such that each repair $(I',S')$ is minimal and $\tau$-constrained.  A naive approach  evaluates all candidate senses against all $x \in \Pi_X$, leading to an exponential number of solutions, which is not feasible in practice.  

In this section, we introduce a two-step greedy approach.
\reviseTwo{First, we build an initial solution by considering all senses against each equivalence class $x$, and  selecting the sense that minimizes the number of repairs for each $x \in \Pi_X$ and for each $\phi: X \rightarrow_{syn} A$, leading to a local minimum.  This evaluation greedily selects the sense that leads to the fewest repairs for a given $x$, without considering interactions between the equivalence classes and their senses.}  In the second step, we consider the interactions among pairs of equivalence classes $x, x'$, where $x \in \Pi_X, x' \in \Pi_{X'}$ for $\phi, \phi' \in \Sigma$, respectively, by performing a local search to improve the initial solution w.r.t. minimizing the number of repairs.   \reviseTwo{Our greedy search selects pairs of equivalence classes where the modeled data distributions share the greatest distance, indicating a larger amount of conflict.}   We model the interactions among all $\phi, \phi' \in \Sigma$ that share a common consequent attribute by constructing a \emph{dependency graph}, where nodes represent equivalence classes and edges represent a shared consequent attribute. We traverse this dependency graph to refine the initial solution.  We evaluate the tradeoff between updating the initial sense assignment and a new assignment to $x$ by estimating the number of necessary data and ontology repairs to achieve alignment between $x$ and $x'$.  We visit nodes $x'$ within a neighbourhood of the current node $x$ as defined via a distance function.  We use the \emph{Earth Mover's Distance (EMD)} to quantify the amount of work needed to transform the data distributions modeled in $x$ to the distribution represented in $x'$~\cite{Rubner00}.  We visit nodes $x'$ where the EMD values between $x$ and $x'$ are largest, indicating regions with the greatest amount of conflict, while avoiding exhaustive enumeration of all candidate assignments.  
We first describe our greedy algorithm to generate an initial solution, followed by local neighbourhood refinement.

\subsection{Computing an Initial Assignment} \label{sec:greedy}
For $\phi: X \rightarrow_{syn} A$,  $x \in \Pi_X$, we compute $\lambda_x$, an initial sense assignment for each $x$.  To compute a minimal repair, we seek senses $\lambda_x$ containing as many values as possible (maximal overlap) with values in $x$.  Let $t_u, v_u$ represent the number of distinct tuples and values, respectively, that are not covered by $\lambda_x$.   Let $\vec{r}_{\lambda_{x}} =  (t_u, v_u)$ be a vector representing the number of necessary tuples and values that must be updated in $x$ under sense $\lambda$ to satisfy $\phi$.  We seek a sense $\lambda'$ such that $\vec{r}_{\lambda_{x}'} \prec  \vec{r}_{\lambda_{x}}, \forall \lambda_x$. 
\eat{i.e., the repair $\vec{r}_{\lambda_{x}'}$ dominates $\vec{r}_{\lambda_{x}}$, for all senses. }

The enumeration of all candidate senses to $x$ is practically inefficient.  We take a greedy approach that finds a $\lambda_x$ with maximal overlap with values in $x$.  We first construct an index of all senses containing each value $v \in x$, denoted by $sset(v)$.  For example, Figure~\ref{fig:ecg}(c) shows $sset(c_3) = \{\lambda_1, \lambda_4\}$.  To maximize the coverage of values, we naturally seek values with maximal frequency.  However, given that values in $x$ may contain errors, we seek an ordering of the values in $x$ that is robust to outliers.  The \emph{Median Absolute Deviation (MAD)} measure quantifies the variability of a sample, being more resilient to outliers than the standard deviation~\cite{rc1993}.  \eat{Since we assume that errors are normally small in number, the deviations from a small number of outliers are less relevant.}  Let $f(v)$ denote the frequency of value $v$. For a univariate dataset $\{f(v_1), f(v_2), \ldots f(v_n)\}$, MAD is defined as the median of the absolute deviations from the median, i.e., MAD($x$) = median$(|f(v_i) - f(\bar{(v)})|)$, where $f(\bar{(v)})$ is the median of all values in $x$.   We sort the values in $x$ according to their decreasing MAD scores.  We iteratively search for a sense $\lambda_x$ containing as many values as possible from $x$ with the highest MAD scores.  \reviseTwo{We use $k'$ as a (decreasing) counter to find such a sense covering $k'$ values in $x$, where $k'$ is initialized to $n$, the number of unique values in $x$.}  \eat{, i.e., starting with $k = |MAD($x$)|$.}  For each value $v_i$, we compute the intersection of their respective $sset(v_i)$.  If there exists a non-empty intersection of senses, the algorithm ends, and returns these senses as an assignment for $x$, denoted as $potential\_set(x)$. We select the sense $\lambda_x$ covering the largest number of tuples in $x$.   Algorithm~\ref{pc:greedyassignment} provides the details.

\begin{figure}[!t]
\noindent\begin{minipage}[t]{.5\textwidth}
\begin{algorithm}[H]
\centering
\caption{Initial$\_$Assignment($x$, $S$, $\Lambda$)}\label{pc:greedyassignment}
  \footnotesize
    \begin{algorithmic}[1]
        \STATE let $sset(v) \leftarrow \emptyset$ 
        \FOR{each $\lambda \in \Lambda$} \label{line:indexstart}
            \FOR{each value $v \in \lambda$}
                \STATE $sset(v) \leftarrow sset(v) \bigcup \lambda$
            \ENDFOR
        \ENDFOR \label{line:indexend}
        \STATE $k' \leftarrow MAD(x)$ \label{line:m}
        \STATE $potential\_set(x) \leftarrow \emptyset$ 
        \WHILE{$potential\_set(x)$ is $\emptyset$}
            \FOR{each top-$k'$ $\{f(v_1), ..., f(v_n)\} \in MAD(x)$}
                \STATE $\Lambda' \leftarrow sset(v_1) \bigcap ... \bigcap sset(v_n) $
                \IF{$\Lambda' \neq \emptyset$}
                    \STATE $potential\_set(x) \leftarrow potential\_set(x) \bigcup \Lambda'$
                \ENDIF
            \ENDFOR
            \STATE $k' \leftarrow k'-1$ \label{line:endwhile}
        \ENDWHILE
        \FOR{each sense $\lambda \in potential\_set(x)$} \label{line:ppset}
            \STATE $cover(\lambda) \leftarrow$ the tuple coverage of $\lambda$ over $x$
        \ENDFOR 
        \STATE $\lambda_x \leftarrow \lambda \in cover(\lambda)$ with maximal tuple coverage \label{line:lambdax}
        \RETURN $\lambda_x$ \label{line:initialresult}
    \end{algorithmic}
\end{algorithm}
\end{minipage}%
\hfill
\begin{minipage}[t]{0.46\textwidth}
\begin{algorithm}[H]
    \centering
     \caption{Local$\_$Refinement($G$, $x$, $\Lambda$)}\label{pc:refine}
    \footnotesize
   \begin{algorithmic}[1]
        \STATE $u_1 \leftarrow $ BFS($G$) 
        \FOR{each vertex $u_2$ connected to $u_1$}
            \IF{$w(u_1,u_2) > \theta$} \label{line:theta}
                \STATE compute $min\{$\cost $\}$ repair \label{line:mincost}
                \IF{sense reassignment for $u_2$ to $\lambda'$} 
                    \IF{$w'(u_1,u_2) < w(u_1,u_2)$}  
                        \STATE $\Lambda \leftarrow \Lambda(I) \setminus \{u_2, \lambda\}$
                        \STATE $\Lambda \leftarrow \Lambda(I) \bigcup \{u_2, \lambda'\}$ \label{line:reassign}
                    \ENDIF
                \ENDIF
            \ENDIF
        \ENDFOR
        \STATE continue until all all $u_1$ visited \label{line:stop}  
        \RETURN $\Lambda$ \label{line:refineresult}
    \end{algorithmic}
\end{algorithm}
\end{minipage}
\vspace{-0.25cm}
\end{figure}

\begin{example}\label{ex:ecg}
Figure~\ref{fig:ecg}(a) shows a data instance $I$ for two synonym OFDs $\phi_1: A \rightarrow C$ and $\phi_2: B \rightarrow C$.  Figure~\ref{fig:ecg}(b) shows seven senses and the corresponding synonym values for each sense.  Let $\lambda_{x_3}$ represent equivalence class $x_3 = \Pi_{B=b_2}$. Figure~\ref{fig:ecg}(c)) shows the index of senses for each value for all $x$.  We compute the MAD for values in $x_3$, i.e.,  $c_1 - c_4$, and determine that the ranked ordering is $\{c_4, c_2, c_1, c_3\}$.  We intersect the senses in $sset(c_2)$ and $sset(c_4)$, resulting in $\lambda_2$ as the initial sense for $x_3$. 
\end{example}

\subsection{Local Refinement} \label{sec:refinement}
We now discuss how we model interactions between $x$ and $x'$, and address conflicting sense assignments that can lead to an increased number of repairs.

\subsubsection{Modeling Interactions} \label{sec:interactions}

Let $\Omega_{x,x'}(\lambda) = \{t | t \in (x \cap x')\}$ represent the set of overlapping tuples between two equivalence classes $x, x'$ from $\phi, \phi'$, respectively, \reviseTwo{that share the same consequent attribute}, interpreted w.r.t. sense $\lambda$.  To quantify conflict among $\phi, \phi'$, we need to interpret tuple values in the intersection of $x$ and $x'$ under the assigned sense. That is, conflicts may arise when there is a difference in the interpretation of the tuple values.  Let $x_{\lambda}$ represent the set of values in $x$ when interpreted under sense $\lambda$.  Given synonym relations, assume that for each sense $\lambda$, there exists a canonical value representing  all  equivalent values in $\lambda$.  Let $D(x_{\lambda})$ represent the distribution of values in $x_{\lambda}$, where values in $x$ that are in $\lambda$ are represented by the canonical value.   If $x$ and $x'$ are assigned the same sense $\lambda$, then tuples in their intersection, ($\Omega_{x,x'}(\lambda)$), share the same distribution,  $D(\Omega_{x,x'}(\lambda))$.   We use \Omegalambda \ to denote overlapping tuples in $x$ and $x'$ when it's clear from the context.

\begin{example}
Let $x_2 = \Pi_{A=a_1}$ and $x_3  =\Pi_{B=b_2}$, highlighted in blue and red, respectively, in Figure~\ref{fig:ecg}(a).  If we assign sense $\lambda_1$ to both classes, then the tuples in $\Omega_{x_2,x_3}(\lambda_1) = \{t_3, t_4, t_5, t_6\}$  share the same distribution, as shown via light grey bars in Figure~\ref{fig:ecg}(d), assuming that $c_2$ is the canonical value for $\lambda_1$. \eat{Assume $c_2$ is the canonical value for $\lambda_1$, then $\Omega$ of $x_2$ and $x_3$ shares the same distributions.  If we assign $\lambda_2$ to $x_2$ and $\lambda_1$ to $x_3$,  Figure~\ref{fig:ecg}(d) shows the two different distributions of values in $\Omega$.}
\end{example}

The above example demonstrates that when two distributions,  $D$(\Omegalambda) and $D(\Omega(\lambda'))$, are equivalent for $\lambda = \lambda'$, the \emph{outliers}, i.e., values not captured by $\lambda$, translate into necessary repairs to achieve consistency w.r.t. tuples in $\Omega$. In the above example, $c_4$ is the outlier, and we can resolve the inconsistency by updating $c_4$ to $c_2$ or by adding $c_4$ to $\lambda_1$.
\eat{In this case, the two interacted equivalence classes share the same outliers in $\Omega$. For example, as shown in Figure~\ref{fig:ecg}(d), consider assigning $\lambda_1$ to both $x_2$ and $x_3$. Assume $c_2$ is the canonical value for $\lambda_1$, fei: already said previously, redundant with the example.} In contrast, when $\lambda \neq \lambda'$, tuples in \Omegalambda and $\Omega(\lambda')$ will share different outliers.  For example, consider assigning $\lambda_1$ to $x_2$, and $\lambda_2$ to $x_3$.  The distributions for $\Omega_{x_2,x_3}(\lambda_1)$ and $\Omega_{x_2,x_3}(\lambda_2)$ are shown via the light grey bars and dark grey bars, respectively, in Figure~\ref{fig:ecg}(d).  Assuming $c_2$ is the canonical value for each sense, the outliers are $c_1$ and $c_3$ for $\lambda_2$, and $c_4$ for $\lambda_1$.  Let \rhoomega \ represent the set of outliers w.r.t. sense $\lambda$, i.e., the unique values in \Omegalambda \ not covered by  $\lambda$, \rhoomega = $\{v | v \in \Omega_{x,x'}(\lambda), v \not \in \lambda\}$.   We can resolve these outliers by: (1) updating the sense of $x$ from $\lambda$ to a new $\lambda'$; (2) adding outlier values \rhoomega \ to $S$; or (3) performing data repairs to update the values in \rhoomega \  to values $v \in \lambda$. 

For a class $x$, we seek a new sense assignment $\lambda'' \neq \lambda$ that includes one or more  values in \rhoomega.  This will cause an increase in the number of data or ontology repairs, as this will require sacrificing at least one value $t[A]$ for attribute $A$ that was covered under $\lambda$, $t \in \{x \setminus \Omega_{x,x'}(\lambda) \}$.  This penalty occurs since if there existed a better sense $\lambda''$ that covered all values in $\rho_{\lambda,x}$  then $\lambda''$ would have been assigned to $x$ during the initialization step (Section~\ref{sec:greedy}).  Hence, we can have $x$ keep $\lambda$ or be re-assigned to $\lambda''$.  However, we must consider the tradeoff of this re-assignment, by incurring a possible penalty (an increased number of repairs) according to the  non-overlapping tuples in $x$.  Alternatively, if $x$ keeps $\lambda$, then we must update the outlier values via data and/or ontology repairs.

\textbf{Achieving Alignment.} 
To evaluate the tradeoff of a sense re-assignment, we define a cost function \cost \ that sums the number of data and ontology repairs to align  $D$(\Omegalambda) and  $D(\Omega(\lambda'))$.  For ontology repairs, to re-align $D$(\Omegalambda) and   $D(\Omega(\lambda'))$, we must add each outlier value in 
\rhoomega $\cup \rho_{\Omega, \lambda'}$ to ontology $S$, via senses $\lambda$ and $\lambda'$, i.e.,  \cost = |\rhoomega| + |$\rho_{\Omega, \lambda'}|$.  Since ontology repairs do not lead to further conflicts, they incur a cost proportional to the number of values added to $\lambda$ and $S$.  For data repairs, we resolve differences by updating tuples $\{t| t \in$ \Omegalambda, $t[A] \in$ \rhoomega $\cup \rho_{\Omega, \lambda'}\}$ to a value from the 
candidate set $(\lambda \cap \lambda')$, which guarantees that the repaired value is covered by both senses.  We set \cost $=  |\mathcal{R}$(\Omegalambda)$| + |\mathcal{R}(\Omega(\lambda'))|$, where $\mathcal{R}$(\Omegalambda) = $\{t |t \in $\Omegalambda, $t[A] \in$ \rhoomega  $\}$, i.e., the number of tuples containing an outlier value.  We select the repair value from the candidate set that minimizes \cost.\   Lastly, we consider sense re-assignment that updates $\lambda$ to $\lambda'$ for $x$.  For $t \in x$, the delta cost, i.e., the number of data repairs to replace $\lambda$ with $\lambda'$ is \cost $=  |\mathcal{R}(x_{\lambda'})| - |\mathcal{R}(x_{\lambda})|$, where $\mathcal{R}(x_{\lambda}) = \{t |t \in x, t[A] \in \rho_{x, \lambda}\}$.  For a pair of classes $x, x'$, we consider all three repair options, and select the option that locally minimizes \cost.

\begin{figure}[!t]
\hspace*{-2cm}
\noindent\begin{minipage}[b]{.3\textwidth}
\centering
\includegraphics[width=3.2cm]{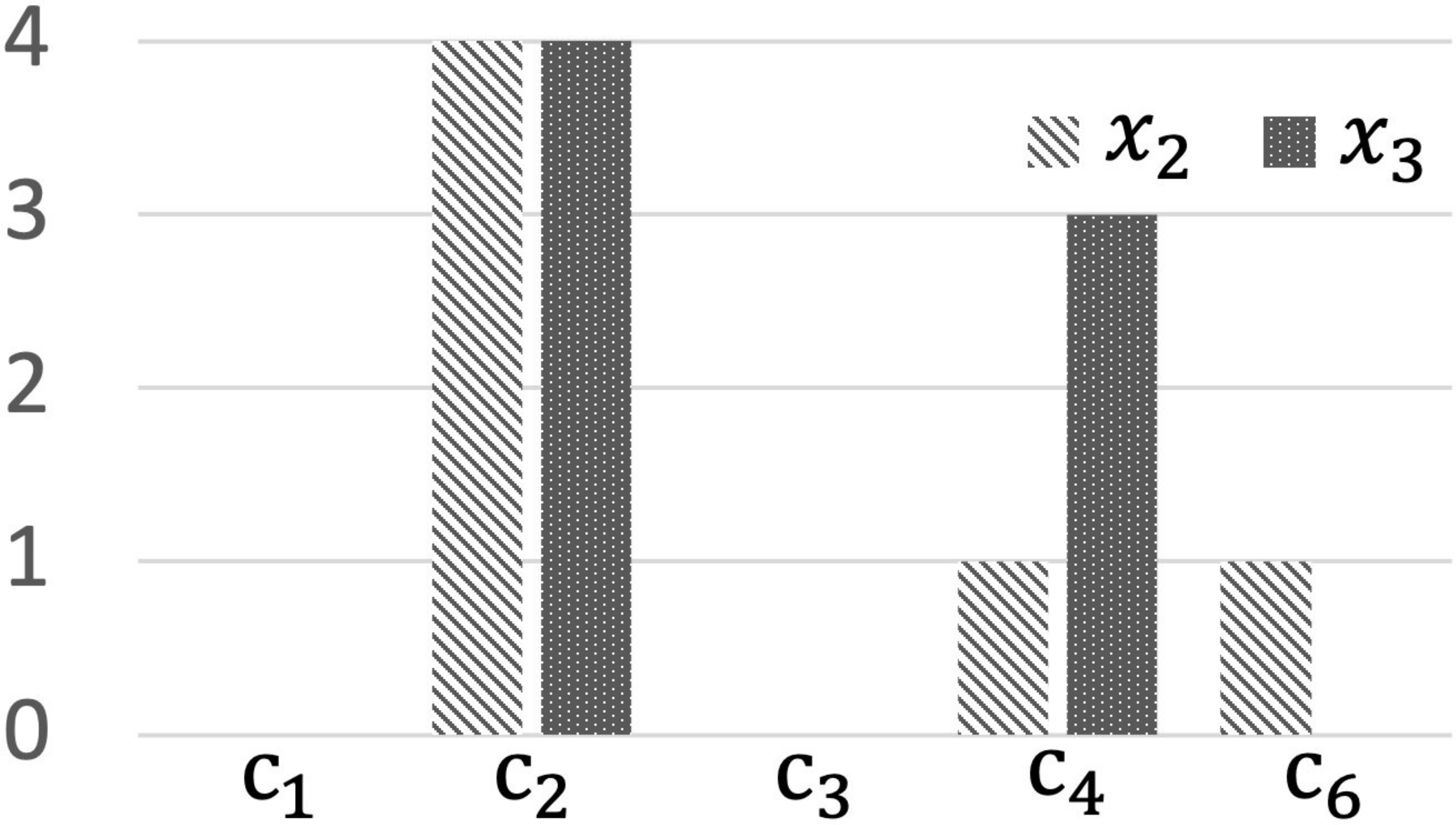}
\caption{$x_{2\lambda_1}$, $x_{3\lambda_1}$ distributions.}
\label{fig:distlambda1}
\end{minipage}
\noindent\begin{minipage}[b]{.3\textwidth}
\centering
\includegraphics[width=3.2cm]{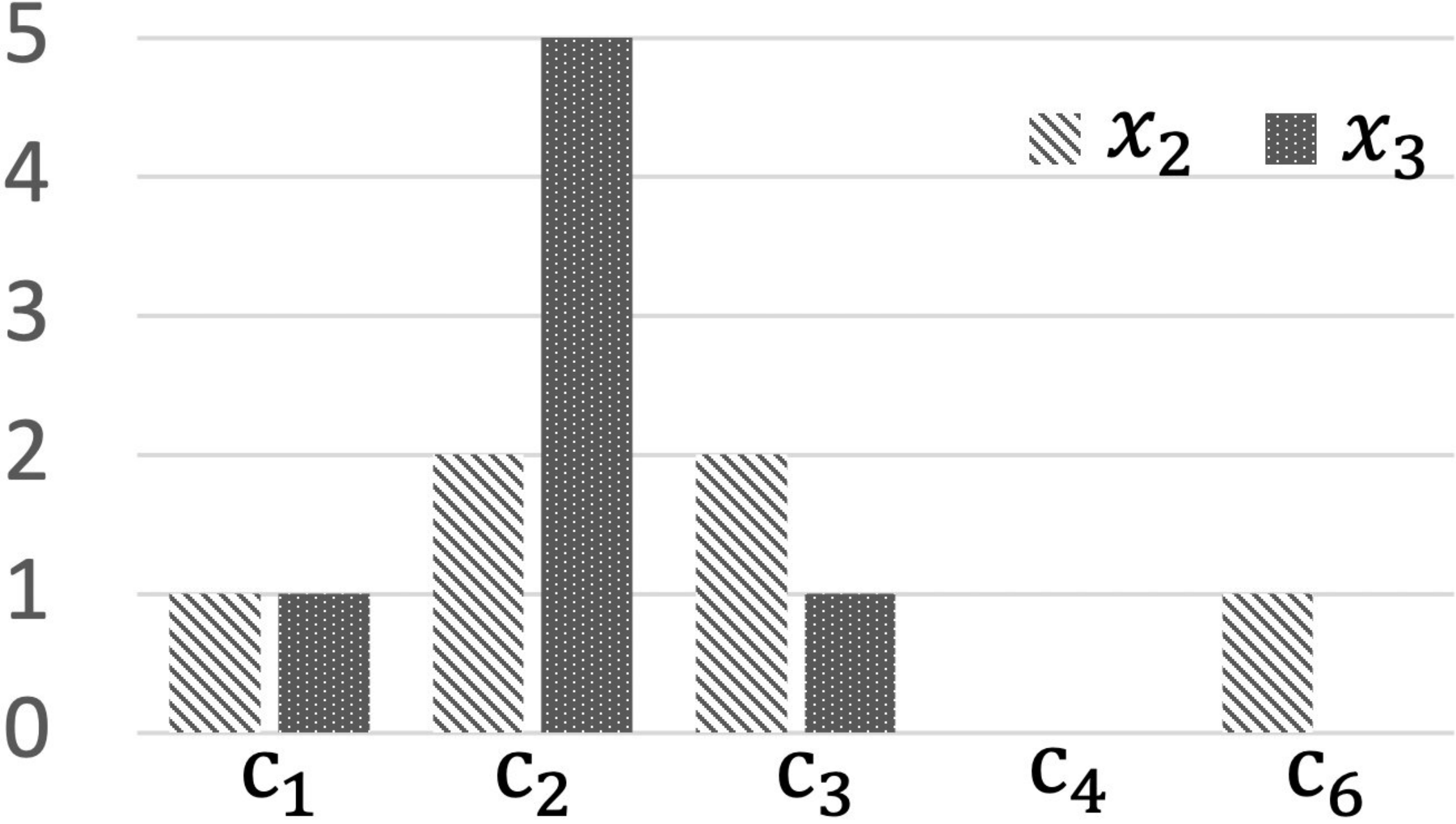}
\caption{$x_{2\lambda_2}$, $x_{3\lambda_2}$ distributions.}
\label{fig:distlambda2}
\end{minipage}
\noindent\begin{minipage}[b]{.3\textwidth}
\centering
\includegraphics[width=6.2cm]{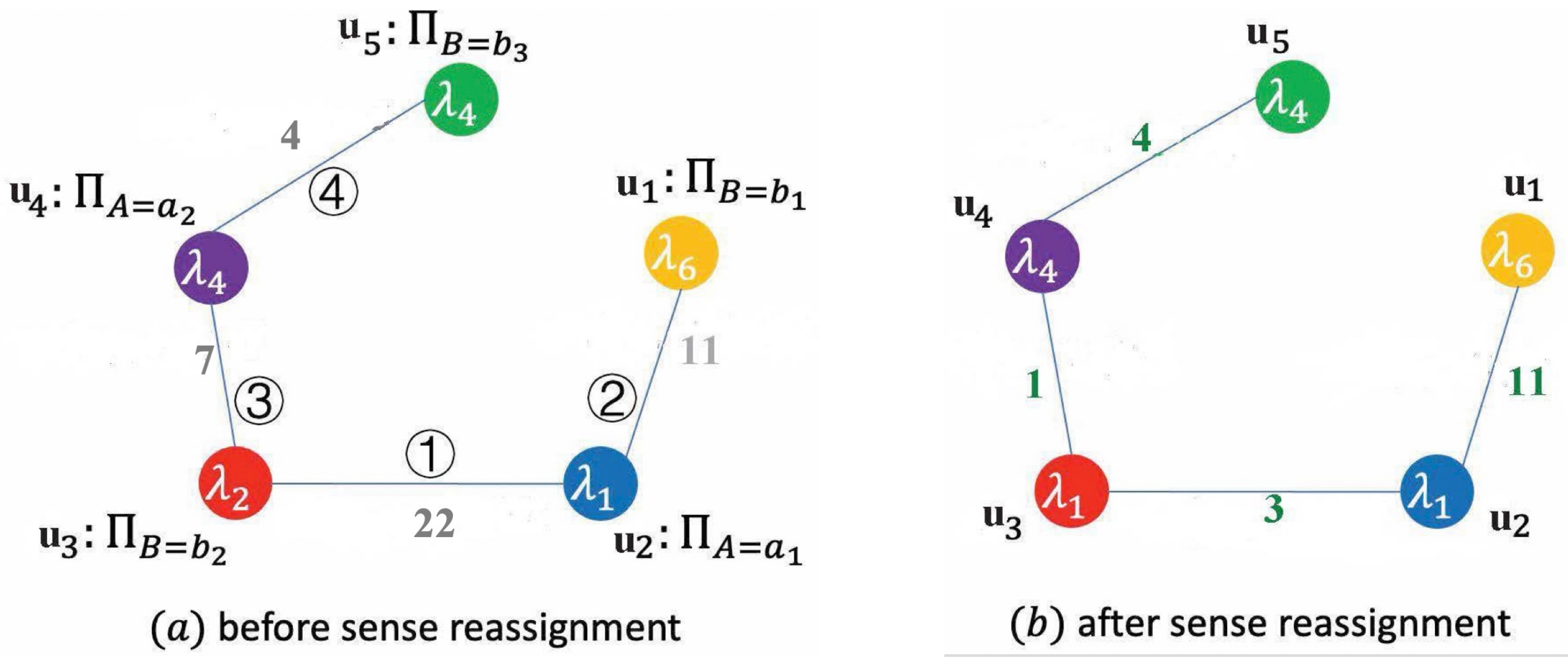}
\caption{Dependency graph.}
\label{fig:dg}
\end{minipage}
\vspace{-0.5cm}
\end{figure}

\begin{example} \label{ex:alignment}
Figures~\ref{fig:distlambda1} and~\ref{fig:distlambda2} show the distribution of values in $x_2$ and $x_3$ with initial senses $\lambda_1$ and $\lambda_2$, respectively. 
Let $c_2$ be the canonical value for $\lambda_1$ and $\lambda_2$. We  have: $\rho_{x_2, \lambda_1} = \{c_4, c_6\}$, $\rho_{x_3, \lambda_1} = \{c_4\}$, and  $\rho_{x_2, \lambda_2} = \{c_1, c_3, c_6\}$, $\rho_{x_3, \lambda_2} = \{c_1, c_3\}$. 
We evaluate ontology and data repair, and two sense assignment options: (i) add $c_4$ to $\lambda_1$, and add $c_1$ and $c_3$ to $\lambda_2$ to eliminate the conflict, leading to \cost = 3; (ii) update three tuples with values \{$c_4$, $c_1$ $c_3$\} to $c_2$, giving \cost = 3; (iii) assign $\lambda_1$ to $x_3$ (from $\lambda_2$), leading to $|\mathcal{R}(x_{3, {\lambda_1}})| - |\mathcal{R}(x_{3, \lambda_2})| = 3- 2 = 1$;
(iv) assign  $\lambda_2$ to $x_2$ where we update $\{c_1, c_3, c_6\}$ to $c_2$ giving \cost  = $|\mathcal{R}(x_{2, \lambda_2)}| - |\mathcal{R}(x_{2, \lambda_1)}| = 4-2 = 2$.  Since \cost \ is minimal with option (iii), we re-assign $\lambda_1$ to class $x_3$. 

\end{example}

\subsubsection{Identifying Candidate Classes}
\label{sec:candclass}
In this section, we discuss how to select pairs of classes $x, x'$ for refinement.  \reviseTwo{Visiting all pairs is not feasible, as we must evaluate $m$ choose 2 pairs of classes, where $m$ equals the total number of equivalence classes across all $\phi \in \Sigma$.}   For classes $x_{\lambda}, x'_{\lambda'}$, we quantify the deviation between their respective distributions, by measuring the amount of work needed to transform $D($\Omegalambda$)$ to $D(\Omega(\lambda'))$ (or vice-versa).  We quantify this work using the Earth Mover's Distance measure, denoted $EMD$~\cite{rtg2000}.  

Let $P$ and $Q$ represent distributions $D(\Omega_{x_1, x_2}(\lambda_1)), D(\Omega_{x_1, x_2}(\lambda_2))$, respectively.  We compute $EMD(P,Q)$ and check whether $EMD(P,Q) > \theta$, for a user-given threshold $\theta$.  Intuitively, if $P$ and $Q$ are sufficiently different (according to $\theta$), then classes $x_1, x_2$ are candidates to consider sense re-assignment to re-align $P$ and $Q$.  Towards finding the sense with minimal repair, if the cost \cost \ of sense re-assignment and the EMD values are lower with the new sense, we proceed with sense re-assignment for $x_1, x_2$.  Otherwise, we keep the initial senses for $x_1$ and $x_2$.

\textbf{Dependency Graph.} 
\label{sec:depgraph}
Let $G = (V,E)$ denote a \emph{dependency graph}, where each vertex $u \in V$ represents an equivalence class $x$, and an edge $e = (u_1, u_2) \in E$ exists when the corresponding classes $(x_1 \cap x_2) \neq \emptyset$.  \reviseTwo{Recall that we define overlap between two classes $x_1$ and $x_2$ when their respective OFDs share a common consequent attribute.  Hence, we only include nodes (and edges) in the dependency graph for OFDs with a common consequent attribute, thereby avoiding enumeration of all pairs of classes.}  We use the terms $u_1, u_2$ to represent classes when it is clear from the context.  Let $w(u_1,u_2)$ denote the weight of edge $e$, representing the EMD value between $u_1$ and $u_2$.  We visit nodes  in a breadth-first search manner, and select nodes $u_1$ with the largest EMD value by summing over all edges containing $u_1$.  This strategy prioritizes classes that require the largest amount of work.  For each visited edge, we check whether $w(u_1, u_2) > \theta$, and if so, check whether a sense re-assignment will reduce the weight, i.e., reduce the EMD value under the new sense.  If so, then the new sense re-aligns the two class distributions and incurs a minimal cost \cost.  The algorithm ends when all vertices have been visited.  Algorithm~\ref{pc:refine} provides the details.

\begin{example}
Figure~\ref{fig:dg}(a) shows the dependency graph corresponding to the equivalence classes in Figure~\ref{fig:ecg}(a), with EMD values as edge weights.  \eat{and the EMD values are shown on the edges.} Suppose $\theta = 10$, and we visit nodes in BFS order of $\{u_2, u_1, u_3, u_4, u_5\}$.  Starting at node $u_2$ (the blue node), we evaluate edge $(u_2, u_3)$ with weight 22. From Example~\ref{ex:alignment}, we update the sense for $u_3$ from $\lambda_2$ to $\lambda_1$ since the new weight $w'(u_3,u_2) < w(u_3,u_2) = 3 < 22$. \eat{ and $\sum e'_{w(x_3)} = 3 + 1 = 4 < e_{w(x_3)} = 22 + 7 = 29$.}  We next consider edge $(u_2, u_1)$ with $w(u_2,u_1) = 11$, with costs to update $u_1$ for data repair, ontology repair and sense reassignment as 1, 1 and 0, respectively.  However, after updating $u_1$ sense $\lambda_6$ to $\lambda_1$, $w'(u_2,u_1)$ does not decrease.  Thus, we do not refine the sense for $u_1$, and keep $\lambda_6$ as is.  \eat{You need to explain this second condition more in your text as I mentioned earlier because it creates questions about why you leave the senses as is, and why this is OK to do.  It also raises questions about EMD and the relationship to our operations to align the distributions.}   The algorithm then visits node $u_3$ and evaluates edge ($u_3$, $u_4$), where after sense reassignment to $\lambda_1$ for $u_3$, we have $w(u_3, u_4) = 1 < \theta$, and there is no further evaluation needed.  Lastly, we visit nodes $u_4, u_5$, where $w(u_4,u_5) = 4 < \theta$, and the algorithm terminates.    Figure~\ref{fig:dg}(b) shows the final sense assignments.
\end{example}

\subsection{Algorithm}
\label{sec:sensealg}
Algorithm~\ref{pc:algorithm} presents the details of our sense assignment algorithm.  We first compute an initial assignment for every equivalence class $x$.  We construct the dependency graph $G$, and compute the $EMD$ between overlapping classes ($u_1, u_2$) as edge weights.  We visit nodes in decreasing order of their $EMD$ values by summing over all corresponding edges.  For example, Figure~\ref{fig:ecg}(d) shows that we visit $u_2$ first since the EMD value of $u_2$ = 22 + 11 = 33, and of $u_3$ = 22 + 7 = 29.  Furthermore, we visit edge $(u_2, u_3)$ first since $w(u_2,u_3) = 22$ is the largest edge weight of $u_2$. We refine the sense for each equivalence class based on Algorithm~\ref{pc:refine}. Lastly, the algorithm returns the final sense assignment.

\noindent \textbf{Complexity.} Computing an initial assignment takes time O($|\Pi_{X}| \cdot m$) to evaluate all $|\Pi_{X}|$ equivalence classes, and in the worst case, the total $m$ senses.  Traversing the dependency graph is in the same order as BFS, taking time $O((|V|+|E|) \cdot m)$, and to evaluate senses for each visited node.     

\begin{figure}[!t]
\begin{minipage}[t]{0.46\textwidth}
\begin{algorithm}[H]
    \centering
    \caption{Sense$\_$Assignment($I$, $\Sigma$, $S$, $\theta$)}\label{pc:algorithm}
    \footnotesize
    \begin{algorithmic}[1]
       \STATE  $\Lambda \leftarrow \{\emptyset, \emptyset\}$ 
        \FOR{each $\phi \in \Sigma$} \label{line:initialstart}
            \FOR{each $x \in \Pi_X(I)$}
                \FOR{each  $\lambda \in \Lambda$}
                    \STATE $\lambda_x \leftarrow$ Initial$\_$Assignment($x$, $S$, $\Lambda$)
                    \STATE $\Lambda \leftarrow \Lambda \bigcup \{x, \lambda_x\}$ \label{line:initialend}
                \ENDFOR
            \ENDFOR
        \ENDFOR
        \STATE  $G$ = CreateDepGraph($I, \Sigma, S, \Lambda$)  
        \STATE \textbf{sort}($G$, DSC)  Sort vertices in DSC order EMD edge weights
        
       \WHILE{($u_1$ = BFS($G$) and EMD($u_1$) > $\theta$)}
        \STATE $\Lambda \leftarrow$  Local$\_$Refinement($G$, $x$, $\Lambda$) \label{line:refine}
        \ENDWHILE
        \RETURN $\Lambda$ \label{line:finalresult}
    \end{algorithmic}
\end{algorithm}
\end{minipage}
\hfill
\noindent\begin{minipage}[t]{.5\textwidth}
\begin{algorithm}[H]
\centering
   \caption{Ontology$\_$Repair ($I$, $\Sigma$, $S$, $b$) }\label{pc:repont}
  \footnotesize
    \begin{algorithmic}[1]
        \STATE $Cand(S) \leftarrow \emptyset$ 
        \STATE $\mathcal{P(S)} \leftarrow \{\emptyset, ..., \emptyset \}$ 
        \FOR{each $\phi \in \Sigma$}
            \FOR{each $x \in \Pi_X(I)$}
                \IF{$t_{e} \in x$ and $t_{e}[A] \not \in S$}
                    \STATE $Cand(S) \leftarrow Cand(S) \bigcup t_{e}[A]$
                \ENDIF
            \ENDFOR
        \ENDFOR
        \STATE $k \leftarrow 1$
        \STATE Generate $\mathcal{V}_k$ as clusters of size $k$, with values $v \in Cand(S)$
        \STATE $L_k \leftarrow \{v_k | v_k \in {\mathcal{V}_k} \}$
        \WHILE{$L_k \neq \emptyset$}
            \STATE $\delta(v_k)$ = compute data repairs for $v_k$ 
            \STATE $L_k \leftarrow$ top-$b$ nodes with smallest $\delta(v_k)$ to expand
            \STATE $v_k* = \{v_k |$ min $\delta(v_k)\}$
            \STATE $\mathcal{P(S)} \leftarrow (S', I')$ using $v_k*$ 
            \STATE $k \leftarrow k+1$
        \ENDWHILE
        \RETURN $\mathcal{P(S)}$
    \end{algorithmic}
\end{algorithm}
\end{minipage}%
\vspace{-0.35cm}
\end{figure}

\section{Computing Repairs}\label{sec:ontrep}
After each equivalence class $x$ is assigned a sense, we have an interpretation from which to identify errors and to compute repairs.  We describe our repair algorithm that first evaluates ontology repair candidates, i.e., values that are in $I$ but not in $S$ under the chosen sense.  Second, assuming a set of ontology repairs, we discuss how the remaining data violations are modeled and repaired.

\subsection{Ontology Repairs}
\label{sec:beamsearch}
We consider ontology repairs that add new values $v \in I$, but $v \not \in S$.  Let $Cand(S)$ represent all candidate ontology repair values.  We iterate through $Cand(S)$ by considering repairs of size $k$, $k = \{1, 2, ..., |Cand(S)|\}$.  For ease of presentation, let $\mathcal{V}_k$ represent the set of candidate ontology repairs of size $k$.  For each $\mathcal{V}_k$, we compute the number of data repairs $\delta_k$ needed to achieve consistency w.r.t. $\Sigma$, and select data repairs with a minimum number of updates for each $k$.
To generate a Pareto-optimal set of repairs, we select $(S', I')$ that are minimal for each $k$ value. 


Generating all $k$-combinations of solutions has exponential complexity, taking $O(2^{|Cand(S)|})$.  We propose a greedy strategy to selectively consider the most promising solutions.  We use the \emph{beam search} strategy, a heuristic optimization of breadth-first search, that expands  the top-$b$ most promising nodes, for a beam size $b$.  The parameter $b$ can be tuned according to application requirements; in our case, we select $b$ by maximizing the probability of selecting the value from a random sequence.  The \emph{Secretary Problem} addresses this issue, where the objective is to select one secretary from $w$ candidates, and applicants are considered in some random order (all orders are equally likely)~\cite{ferguson1989}.  Applicants are chosen immediately after being interviewed, and decisions are non-reversible.   Prior work has shown that it is optimal to interview $b = \left \lfloor \frac{w}{e} \right \rfloor$ candidates, where $e$ is the exponential constant.  This strategy selects the best candidate with probability $\frac{1}{e}$, with a competitive ratio $e$.  In our case, we set $w = |Cand(S)|$.  

We organize the candidate repairs as a set-containment lattice,  where nodes at level $k$ represent ontology repairs of size $k$, e.g., at level 1, we consider single value repairs. Let $v_k$ represent a node at level $k$, i.e., the repairs in $v_k \in \mathcal{V}_k$.  Each node $v_k$ is created by augmenting a node at level $k-1$ with a single repair value.  We begin the search at a node $v_1$ at level $k=1$.  As we visit a node $v_k$, we compute the minimum number of data repairs for the current candidate ontology repair $v_k$, i.e., assuming $S' = \{S \cup v_k\}$.  We select the top-$b$ nodes at each level $k$ with the minimum number of data repairs, to further explore at the next level $k+1$.  We continue this traversal at each subsequent level until we have $I' \models \Sigma$ w.r.t. $S'$, or until we reach the leaf level.
Algorithm~\ref{pc:repont} presents the details.  \eat{Figure~\ref{fig:ontsearch} shows an example lattice with \tbf ontology repair candidates, and beam size $b = \left \lfloor \frac{\tbf}{e} \right \rfloor = \tbf$.}  

\begin{table}[t]
\begin{minipage}[tb]{0.45\linewidth}
\scriptsize
		\caption{Sample clinical trials data.}\label{tab:subset}
			\vspace{-0.35cm}
   \begin{tabular}{ | l | l | l | l | l | l |}
    \hline
    \hspace{-2 mm} \textbf{id} \hspace{-2 mm} & \textbf{CC}    & \textbf{CTRY} & \hspace{-2 mm} \textbf{SYMP} & \hspace{-2 mm} \textbf{DIAG}  & \hspace{-2 mm} \textbf{MED}  \\
    \hline \hline
    \hspace{-2 mm} $t_8$ \hspace{-2 mm} &   US  & USA & headache & hypertension  & cartia \\
	\hspace{-2 mm} $t_9$ \hspace{-2 mm} &   US  & USA & headache & hypertension & ASA \\
	\hspace{-2 mm} $t_{10}$ \hspace{-2 mm} & US & America & headache & hypertension & tiazac \\
	\hspace{-2 mm} $t_{11}$ \hspace{-2 mm} & US & Uni. States & headache & hypertension & adizem \\
    \hline
    \end{tabular}
    
    \centering
     \rule{0.3\textwidth}{0pt}
    \includegraphics[width=4.5cm]{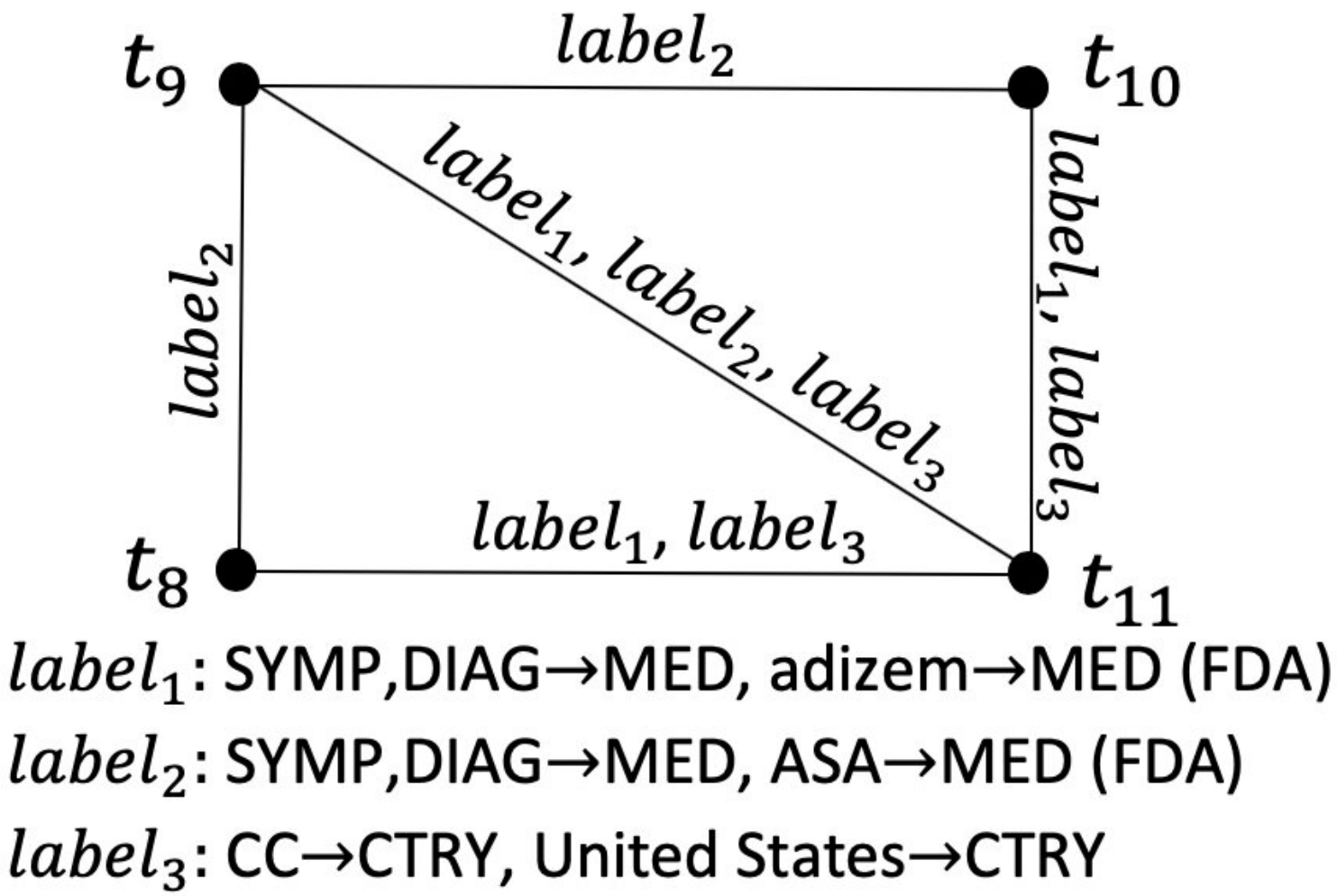}
    	\vspace{-0.25cm}
	\captionof{figure}{Conflict graph for Table~\ref{tab:subset}.}
	\label{fig:cg}
\end{minipage}
\hspace{0.3cm}
\begin{minipage}[tb]{0.45\linewidth}
\centering
\scriptsize
	\caption{Sample repairs.}\label{table:ontrep}
		\vspace{-0.2cm}
	    \begin{tabular}{|c|c|c|c|c|}
    		\hline
    		Ont. Repair & $|S,S'|$ &  Conflict Edges & $C_{2opt}$ & $\delta_P$ \\
    		value (sense) & & & & \\
    		\hline
    		$\emptyset$ & 0 & $(t_8, t_{9})$, $(t_{8}, t_{11})$, $(t_{9}, t_{10})$ & $t_9, t_{11}$ & 4 \\
    		 & & $(t_{9}, t_{11})$, $(t_{10}, t_{11})$ & & \\
    	    \hline
    		ASA (FDA) & 1 &  $(t_8, t_{11})$, $(t_9, t_{11})$, $(t_{10}, t_{11})$ & $t_{11}$ & 2 \\ 
    		\hline
    		adizem (FDA) & 1 & $(t_8, t_{9})$, $(t_{8}, t_{11})$, $(t_{9}, t_{10})$ & $t_9, t_{11}$ & 4 \\
    		& & $(t_{9}, t_{11})$, $(t_{10}, t_{11})$ & & \\ 
    		\hline
    		United States  & 1 & $(t_8, t_{9})$, $(t_{8}, t_{11})$, $(t_{9}, t_{10})$ & $t_9, t_{11}$ & 4 \\
    		 & & $(t_{9}, t_{11})$, $(t_{10}, t_{11})$ & & \\
    		\hline
    		adizem (FDA) & 2 & $(t_8, t_{9})$, $(t_9, t_{10})$, $(t_{9}, t_{11})$ & $t_{9}$ & 2 \\
    		United States & & & &\\ 
    		\hline
    		... & ... & ... & ... & ... \\ 
    		\hline
	    \end{tabular}
\end{minipage}
	\vspace{-0.6cm}
\end{table}

\eat{
\begin{minipage}[ht]{.75\linewidth}
  \centering
  \rule{0.3\textwidth}{0pt}
    \includegraphics[width=6.5cm]{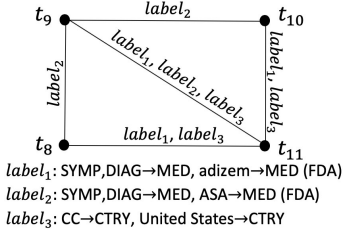}
	\captionof{figure}{Conflict graph for Table~\ref{tab:subset}.}
		\label{fig:cg}
\end{minipage}
}
\subsection{Approximating Minimum Data Repairs}
\label{sec:minimallyrepair}
Given $v_k$ to derive $S'$, we compute the necessary data repairs $I'$ such that $dist(I,I') \leq \tau$.
Computing a minimal number of data repairs to $I$ such that $I' \models \Gamma$, for a set of FDs $\Gamma$, is known to be NP-hard~\cite{MinCost1}.  Since OFDs subsume FDs, i,e., $\Gamma \subseteq \Sigma$, this intractability result, unfortunately, carries over to OFDs.  Beskales et al., show that the minimum number of data repairs can be approximated by upper bounding the minimum number of necessary cell changes~\cite{BIGG13}.  In our implementation, we adapt their {\sf RepairData} algorithm that is shown to compute an instance $I'$, by cleaning tuples one at a time, such that the number of changed cells is at most $P = 2 \cdot$min$\{|Z|, |\Sigma|\}$, where $|Z|$ is the number of (unique) consequent attributes in $\Sigma$, and $|\Sigma|$ denotes the number of OFDs.  We seek $P$-approximate,  $\tau$-constrained repairs, where a $P$-approximate $\tau$-constrained repair $(S', I')$ is a repair in $U$ such that $dist(I,I') \leq \tau $, and there is no other repair $(S'', I'') \in U$ such that $(dist(S, S''), P \cdot dist(I, I'')) \prec (dist(S, S'), \tau)$.  

\eat{ set of tuples that satisfy $\Sigma$ w.r.t. $S'$ such that the number of satisfying tuples is approximately maximal.}  We transform the data repair problem to the problem of finding a minimum vertex cover, where nodes represent a tuple $t_i$, and an edge $(t_i, t_j)$ represent that $t_i, t_j$ conflict w.r.t. an OFD.  We generate a conflict graph for $I$, w.r.t. $\Sigma, S'$, and compute a 2-approximate minimum vertex cover of the conflict graph~\cite{Garey79}.  We annotate the edges in the conflict graph with: (i) the violated OFD $\phi$; and (ii) the candidate repair and sense. 

\begin{example}\label{example:minirep}
Table~\ref{tab:subset} shows a subset of Table~\ref{tab:cleanexample} where $t_{11}$[CTRY] is updated to \lit{United States}. Consider $\phi_1:$ [CC] $\rightarrow$ [CTRY] with an ontological equivalence between \lit{USA, America}, and $\phi_2:$ [SYMP, DIAG] $\rightarrow$ [MED] w.r.t. the drug (MED) ontology shown in Figure~\ref{fig:ontology}.  If the \lit{FDA} sense is selected for $\phi_2$ (according to Algorithm~\ref{pc:algorithm}), then Figure~\ref{fig:cg} shows the corresponding conflict graph. \eat{\fei{Again, pls. consolidate all examples to continue from your Table 1, Figure 1 example.  Why is the user switching from looking at Figure 4 search space of ontology candidates that are generic ABC, to now specific ones from Table 1.  Provide uniformity and continuity.}}
\end{example}

For each error tuple $t_e$, we modify attribute $A$ by considering candidate repairs from the domain of $A$ when there is sufficient evidence to do so, or to a value $v' \in S'$ such that the values $\{v', a\}$ are synonyms for all $a \in \Pi_{X}$.  As each $t_e$ is repaired, we remove its corresponding nodes and edges, and re-generate the conflict graph in linear time.  The algorithm continues until all tuples $t_e$ have been removed from the conflict graph, and $I'$ is returned.  The algorithm runs in $O(|E| + |I|)$, where $|E|$ is the number of edges in the conflict graph~\cite{BIGG13}.

\begin{example}
Table~\ref{table:ontrep} shows an execution of \ofdclean listing: (i) candidate ontology repairs (value, sense); (ii) number of ontology repairs $|\Delta(S, S')|$; (iii) conflicts in the conflict graph (Figure~\ref{fig:cg}); (iv) $C_{2opt}$: the 2-approximation vertex cover of the conflict graph, indicating the tuples to be repaired; and (v) $\delta_{P}(S',I)$: the upper bound on the minimum number of data repairs to $I$.  Consider an ontology repair to add $\lit{ASA}$ to $S$, which would eliminate edges $(t_8,t_9), (t_9,t_{10})$, and leave only tuple $t_{11}$ to be updated with at most two minimum data repairs.  Comparing all single candidate ontology repairs, this repair is the minimum, involving one and two, ontology and data repairs, respectively. 
\end{example}

\section{Experiments}
\label{sec:eval}
Our evaluation focuses on four objectives: (1) We test \fastofd scalability and performance compared to seven FD discovery algorithms as we scale the number of tuples and the number of attributes. (2) We show the benefits of our optimization techniques to prune redundant OFD candidates. (3) We conduct a qualitative evaluation of the utility of the discovered OFDs. (4) We study the accuracy and performance of \ofdclean.  We show the effectiveness of our sense selection algorithm, and compare \ofdclean against the HoloClean\cite{holoclean} repair algorithm and demonstrate the benefits of ontology repair.

\subsection{Experimental Setup}
\label{sec:expsetup}
Experiments were performed using four Intel Xeon processors at 3.0GHz each with 32GB of memory. Algorithms were implemented in Java.  The reported runtimes are averaged over six executions.  

\noindent \textbf{Datasets.} We use two real datasets, and the U.S. National Library of Medicine Research and WordNet ontologies~\cite{medOntology}.

\noindent \emph{\uline{Kiva}} \cite{Kiva} describes loans issued over two years via the Kiva.org online crowdfunding platform to financially excluded citizens around the world. There are 670K records and 15 attributes, \eat{$|\Sigma| = 8$ and the number of senses $|\lambda| = 2$,}  including loan principal amount, loan activity, country code, country, region, funded time and usage.

\noindent  \emph{\uline{Clinical}} \cite{HKLMW09}: The Linked Clinical Trials (LinkedCT.org) describes clinical patients such as the clinical study, country, medical diagnosis, drugs, illnesses, symptoms, treatment, and outcomes. 
We use a portion of the dataset with 250K records and 15 attributes. \eat{, $|\Sigma| = \tbf, |\lambda| = \tbf$.  }

\reviseTwo{
We use real ontologies to ensure coverage of the (RHS) attribute domain, e.g., the Medical Research ontology covers values in the DRUG attribute of the clinical trials data, while considering some medications have different names across different countries (senses).  We  maximize coverage upwards of 90\%+ coverage for some attributes.  Values not covered by an ontology are candidates for ontology repair.}

\eat{
\begin{table}
\begin{center}\begin{threeparttable}
\centering
\caption{Data characteristics.}
\label{tb:dataset}
\begin{tabular}{l|ccc}
\toprule
{\bf }  & {\bf Kiva}		    & {\bf Clinical}     	\\
\midrule
$N$		&	671,205  	 \tbf	\\
$n$		&	15	     	15	
$|\Sigma|$	&  3 & \tbf	   8\\
$|\lambda|$	&	2   &  \tbf    
\bottomrule
\end{tabular}
\end{threeparttable}\end{center}

\vspace{-5mm}
\end{table}
}

\begin{table}[t]
\begin{minipage}[tb]{0.45\linewidth}
\scriptsize
\caption{Parameter values (defaults in bold) \label{tbl:defaults}}
\vspace{-0.3cm}
\begin{tabular}{ | l | l | l |}
  \hline
  \textbf{Symbol} & \textbf{Description} & \textbf{Values} \\
  \hline
  $|\lambda|$  & \# senses  & 2, \textbf{4}, 6, 8, 10 \\
  \hline
  $err\%$  & error rate & \textbf{3}, 6, 9, 12, 15 \\
  \hline
   $N$ & \# tuples (Million) & 0.2, \textbf{0.4}, 0.6 0.8, 1 (clinical) \\ 
  \hline
  $b$  & beam size & 1, 2, \textbf{3}, 4, 5 \\   
  \hline
  $inc\%$  & incompleteness rate & 2, \textbf{4}, 6, 8, 10 \\
  \hline
   $|\Sigma|$ & \# OFDs &   \textbf{10}, 20, 30, 40, 50 \\
  \hline
\end{tabular}
\end{minipage}
\hspace{0.3cm}
\begin{minipage}[tb]{0.45\linewidth}
\centering
\scriptsize
\begin{threeparttable}
\scriptsize
\caption{Varying $N$ runtimes.}\label{tab:performance}
\begin{tabular}{l|ccccc}
\toprule
\textbf{N (M)}	& \textbf{0.2} 	&  \textbf{0.4}	&  \textbf{0.6} & \textbf{0.8} & \textbf{1} \cr
\midrule
(sec) & 9.3 & 11.8	&  17.1	&  23.8 &	27.2  \cr
\bottomrule
\end{tabular}
\end{threeparttable}
\begin{threeparttable}
\scriptsize
\caption{\ofdclean \ runtimes (m)}\label{tab:runtime}
\begin{tabular}{l|ccccc}
\toprule
\textbf{N (K)}	& \textbf{50} 	&  \textbf{100}	&  \textbf{150} & \textbf{200} & \textbf{250} \cr
\midrule
 & 166 & 175	&  182	&  198 &	217  \cr
\bottomrule
\end{tabular}
\end{threeparttable}
\end{minipage}
\vspace{-0.55cm}
\end{table}

\noindent \textbf{Parameters and Ground Truth.} Table~\ref{tbl:defaults} shows the parameter values and their defaults. We inject errors randomly into the consequent attributes.  Errors are inserted by either changing an existing value to a new value (not in the attribute domain), or to an existing domain value.  We consider the original data values as the ground truth, as we vary the error rate $err\%$.  We specify $\tau$ as a percentage (100\% allows the algorithm to freely change the data).  We set $\tau = 65$\% to balance similarity between $I$ and $I'$, and flexibility to consider new values via ontology repairs. 


\noindent \textbf{Comparative Techniques.}  

\noindent \emph{\uline{FD Mining}}: For the \fastofd comparative experiments, we use the Metanome implementations of existing FD discovery algorithms \cite{PBF+15}.

\noindent  \emph{\uline{HoloClean}} \cite{holoclean} provides holistic data repair by combining multiple input signals (integrity constraints, external dictionaries, and statistical profiling), and uses probabilistic inference to infer dirty values.
We input to HoloClean: (i) a set of denial constraints translated from the given dependencies; (ii) external reference sources such the National Drug Code Directory \cite{NDC}; and (iii) we profile the data to obtain statistical frequency distributions of each attribute's domain.   We compare \ofdclean against HoloClean  since it also considers external information during data repair.


\subsection{\fastofd \reviseTwo{Efficiency}} 
\label{sec:expefficiency}
\eat{Text should have been copied from the CIKM CR version.  It has been condensed in some experiments for space savings.  If there is an important difference, please do highlight.}
\noindent \textbf{Exp-1: Vary $|N|$.}
We vary the number of tuples, and compare \ofdclean \reviseTwo{(with optimizations)} against seven existing FD discovery algorithms: TANE \cite{HKPT98}, FUN \cite{NC01}, FDMine \cite{YH08}, DFD \cite{ASN14}, DepMiner \cite{LPL00}, FastFDs \cite{WGR01}, and FDep \cite{FS99}.  Figure \ref{fig:clinical} shows the running times using $\theta = 5$.  We report partial results for FDMine and FDep as both techniques exceeded the main memory limits.  FDMine returns a much larger number of non-minimal dependencies, about 24x leading to increased memory requirements.   We ran DepMiner and FastFDs using 100K records and report running times of 4hrs and 2.3 hrs, respectively. However, for larger data sizes (200K+ records), we terminated runs for these two techniques after 12 hours. \eat{hence, they do not appear in Figure \ref{fig:clinical}.}  The running times for \fastofd scale linearly with the number of tuples, similar to other lattice traversal based approaches (TANE, FUN, and DFD). The runtime is dominated by data verification of OFDs. 
We found that discovering synonym OFDs incurs an average overhead of 1.8x over existing lattice-traversal FD discovery algorithms due to the inherent complexity of OFDs (which subsume FDs), and the increased number of discovered OFDs, e.g., medications are referenced by multiple names; its generic name and its brand name.

\begin{figure}[tb!]
	\captionsetup[subfloat]{justification=centering}
	\centering
		\subfloat[\small{  Scalability in N.
		}]{\label{fig:clinical}
			{\includegraphics[width=4.2cm,height=3.8cm]{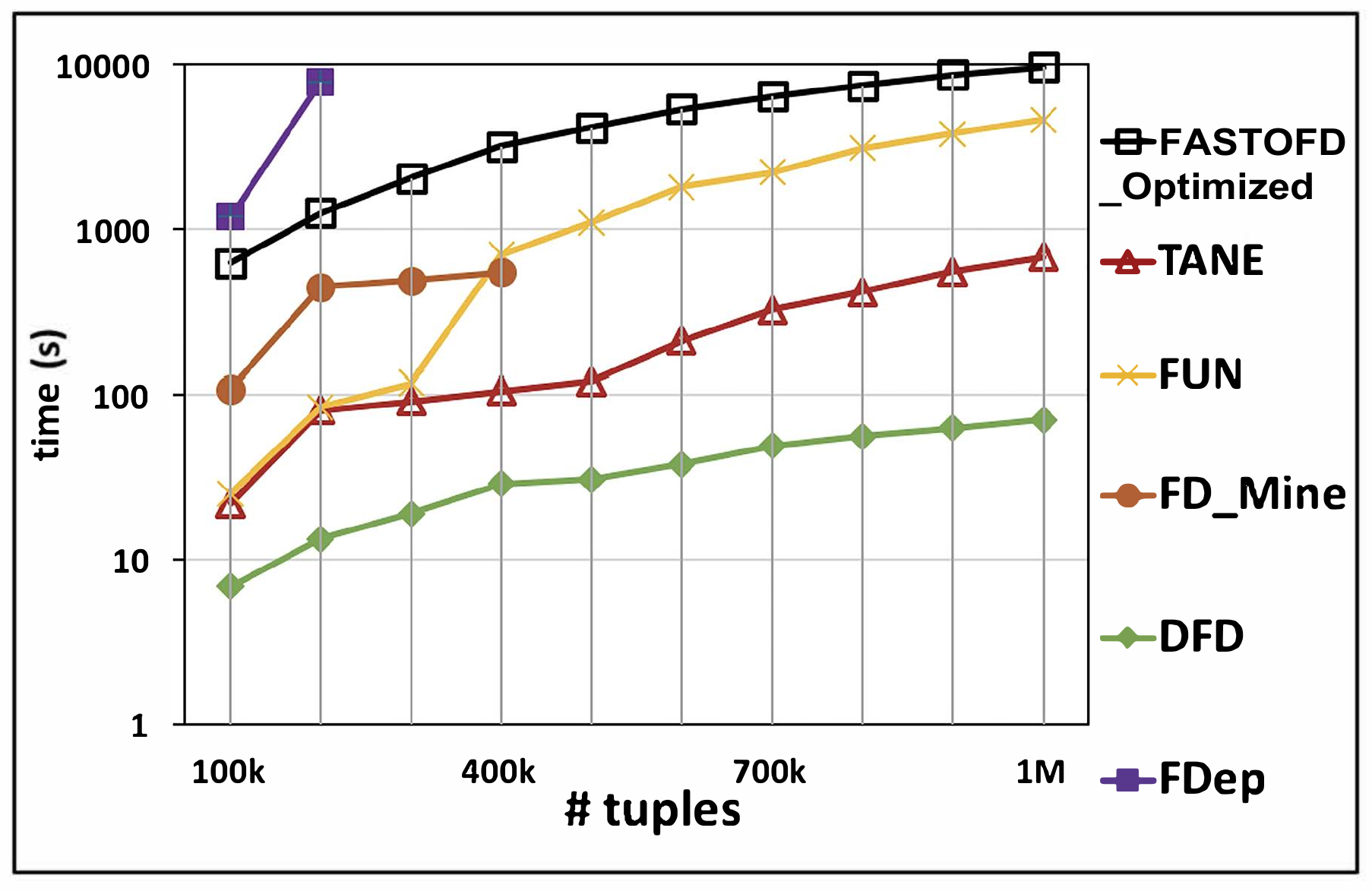}}}
		\hfill\subfloat[
		\small{Scalability in $n$. } 
		]{\label{fig:numattrs}
	{\includegraphics[width=4.2cm,height=3.8cm]{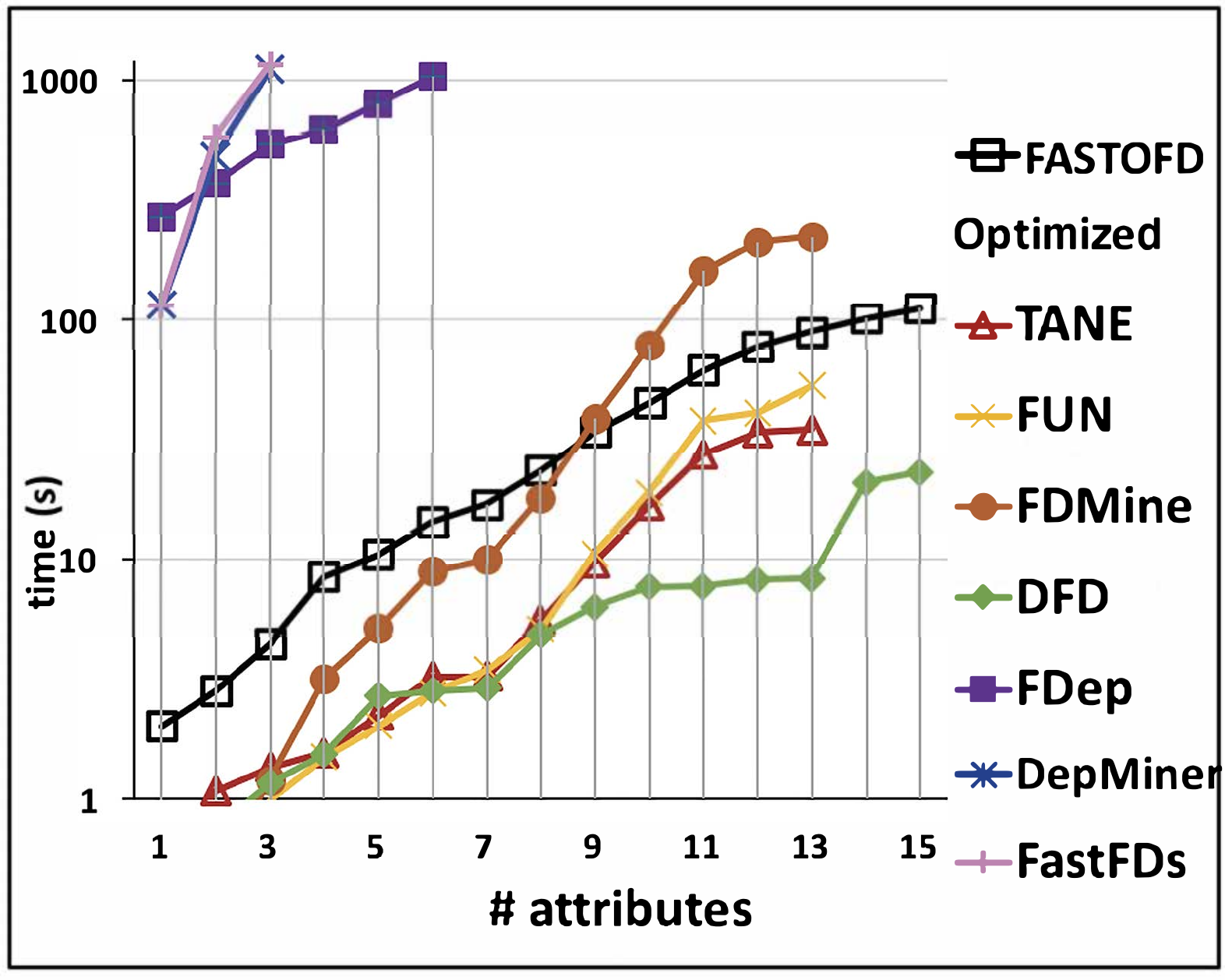}}}
		\hfill\subfloat[
		\reviseTwo{\small{Optimizations. } }
		]{\label{fig:optimizeall}
	{\includegraphics[width=4.2cm,height=3.8cm]{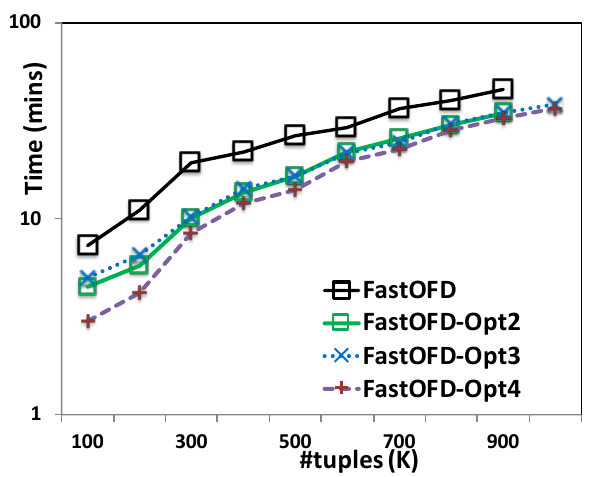}}}

	\caption{\fastofd effectiveness and efficiency (Clinical dataset).}
\end{figure}

\noindent \textbf{Exp-2: Vary $|n|$.}
Figure \ref{fig:numattrs} shows that as we vary the number of attributes ($n$), using $N$ = 100k tuples, and $\theta = 5$, all algorithms scale exponentially with the number of attributes since the space of candidates grows with the number of attribute set combinations.  \fastofd scales comparatively with other lattice based approaches.  We discover 3.1x more dependencies on average, compared to existing approaches, validating the overhead we incur.  In Figure \ref{fig:numattrs}, we report partial results for DepMiner, FastFDs, and FDep (before memory limits were exceeded), where we achieve almost two orders of magnitude improvement due to our optimizations.  \eat{Similar to existing lattice-based approaches, } Our techniques performs well on a smaller number of attributes due to effective pruning strategies that reduce the search space.  

\noindent \reviseTwo{\textbf{Exp-3: Optimizations Benefits.}}  
\reviseTwo{We evaluate optimizations 2, 3, and 4 (Opt-2, Opt-3, Opt-4) from Section~\ref{sec:3_opt}.   We use 1M tuples from the clinical trials dataset.  Figure \ref{fig:optimizeall} shows \fastofd runtime with each of the optimizations, shown in log-scale.}  Opt-2 achieves 31\% runtime improvement due to  aggressive pruning of redundant candidates at lower lattice levels.  For Opt-3, we see a smaller performance improvement (compared to Opt-2) since we reduce the verification time due to the existence of keys, rather than pruning candidates.  We found two key attributes in the clinical data: (1) \texttt{NCTID}, representing the clinical trials.gov unique identifier for a clinical study; and (2) \texttt{OrgStudyID}, an identifier for a group of related studies.  Opt-3 achieves an average 14\% runtime improvement. With more keys, we expect to see a greater improvement in running time.  For Opt-4, we use a set of five defined FDs~\cite{baskaran2017efficient}.  At 100K tuples, we reduce running times by 59\%.  The running times decrease by an average of 27\%.  We expect that the performance gains are proportional to the number of satisfying FDs in the data; an increased number of satisfying FDs lead to less time spent verifying candidate OFDs.  Our optimizations together achieve an average of 24\% improvement in runtime, and our canonical representations are effective in avoiding redundancy.

\noindent \reviseTwo{\textbf{Exp-4: Efficiency over lattice levels.} }
We argue that compact OFDs (those involving a small number of attributes) are the most interesting.  OFDs with more attributes contain more unique equivalence classes.  Thus, a less compact dependency may hold, but may not be very meaningful due to overfitting.  We evaluate the efficiency of our techniques to discover these compact dependencies.  To do this, we measure the number of OFDs and the time spent at each level of the lattice using the clinical trials data.  OFDs discovered at the upper levels of the lattice (involving fewer attributes) are more desirable than those discovered at the lower levels.  Approximately 61\% are found in the first 6 levels (out of 15 levels) taking about 25\% of the total time.  The remaining dependencies (found in the lower levels) are not as compact, and the time to discover these OFDs would take well over 70\% of the total time.   Since most of the interesting OFDs are found at the top levels, we can prune the lower levels (beyond a threshold) to improve overall running times.  

\subsection{\fastofd Effectiveness}
\noindent \textbf{Exp-5: Eliminating False-Positive Data Quality Errors.}
\eat{To quantify the benefits of OFDs versus traditional FDs, we compute the number of syntactically "non-equal" values referring to the same entity.  For example, under traditional FDs, a dependency \texttt{CTRY} $\rightarrow$ \texttt{CC}, with CTRY value \lit{Canada} mapping to CC values \lit{CA}, \lit{CAN}, and \lit{CAD}, would all be considered as errors.  However, under a synonym OFD, these tuples are considered clean, since \lit{CA}, \lit{CAN}, and \lit{CAD} are acceptable representations of \lit{Canada}.    } For each discovered OFD, we sample and compute the percentage of tuples containing syntactically non-equal values in the consequent to quantify the benefits of OFDs versus traditional FDs.  These values represent synonyms.  Under FD based data cleaning, these tuples would be considered errors.  However, by using OFDs, these false positive errors are saved since they are not true errors.  For brevity, we report results here, and refer the reader to \cite{baskaran2017efficient} for graphs.  Overall, we observe that a significantly large percentage of tuples are falsely considered as errors under a traditional FD based data cleaning model.  For example, at level 1, 75\% of the values in synonym OFDs contain non-equal values.  \eat{Only after level 6, we see that the number of tuples containing equal values in the consequent comprise over half of the satisfying tuples in the OFD.  This is somewhat expected, since the dependency becomes more constrained with an increasing number of antecedent attributes.}  By correctly recognizing these `erroneous' tuples as clean, we reduce computation and the manual burden for users to decipher through these falsely categorized errors.  

\subsection{Sense Selection Performance}
We now use the clinical dataset to evaluate the performance of our sense selection algorithm.

\noindent \textbf{Exp-6: Vary $|\lambda|$.}
Figure~\ref{fig:accsense)} shows the accuracy as we increase the number of senses $|\lambda|$.  Since every equivalence class is assigned a sense, we achieve 100\% recall, independent of the number of senses.  Precision decreases with more senses due to more available choices, but still remains above 80\%.  As we evaluate more senses, the runtimes increase linearly, as shown in Figure~\ref{fig:timesense}.  

\noindent \textbf{Exp-7: Vary $err\%$.} Figure~\ref{fig:accerr)} shows that precision declines linearly as the number of errors increases. It becomes more challenging to select the correct sense when multiple senses contain overlapping (erroneous) values. Figure~\ref{fig:timeerr} shows that runtimes increase as more candidate senses and refinements have to be evaluated to align  distributional deviations caused by an increasing number of errors.

\noindent \textbf{Exp-8: Vary $N$.}  We increase the number of tuples in the Clinical dataset up to 1M records.  This increases the number of equivalence classes, but our sense assignment strategy achieves over 90\% precision.  We found that the increase in $N$ did not impact the precision and recall accuracy (figure omitted for brevity).  With more equivalence classes, there is an increased likelihood of overlap, leading to longer runtimes to resolve shared errors, as shown in Table~\ref{tab:performance}.

\begin{figure*}[tb!]
	\captionsetup[subfloat]{justification=centering}
	\centering
		\subfloat[
		\small{Accuracy vs. $|\lambda|$  }]{\label{fig:accsense)}
			{\includegraphics[width=3.2cm,height=2.8cm]{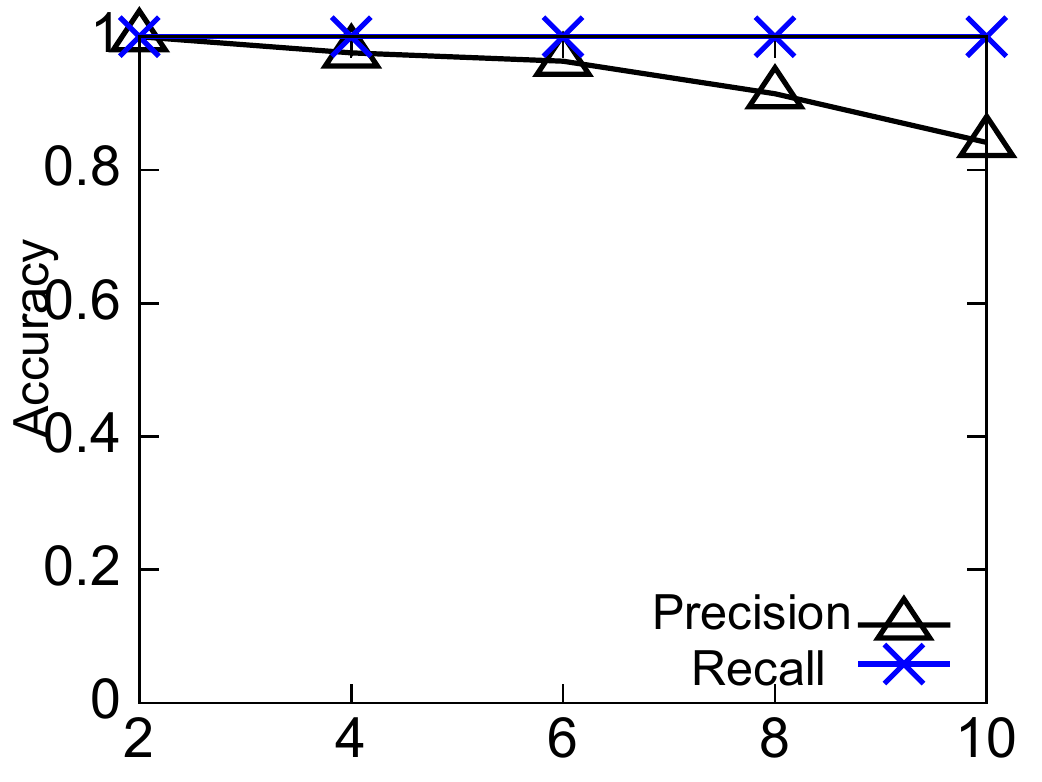}}}
		\hfill\subfloat[
	\small{Time vs. $|\lambda|$ }]
		{\label{fig:timesense}
			{\includegraphics[width=3.2cm,height=2.8cm]{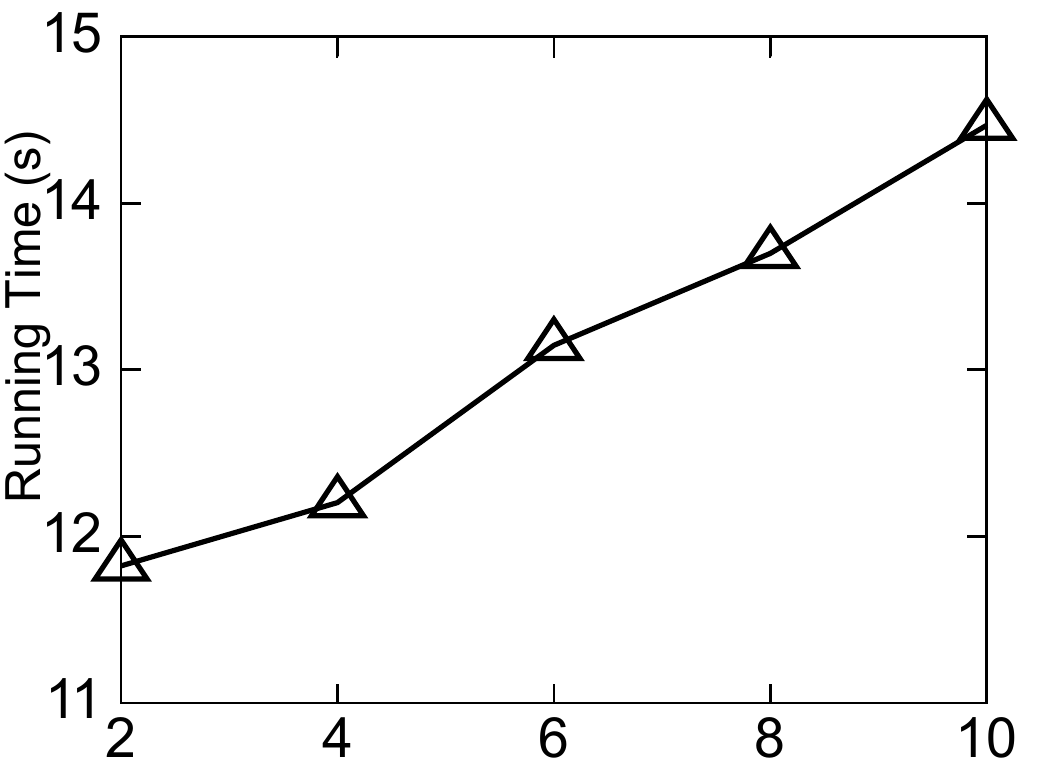}}}	
		\hfill\subfloat[
	\small{Accuracy vs. $err\%$ }]{\label{fig:accerr)}
			{\includegraphics[width=3.2cm,height=2.8cm]{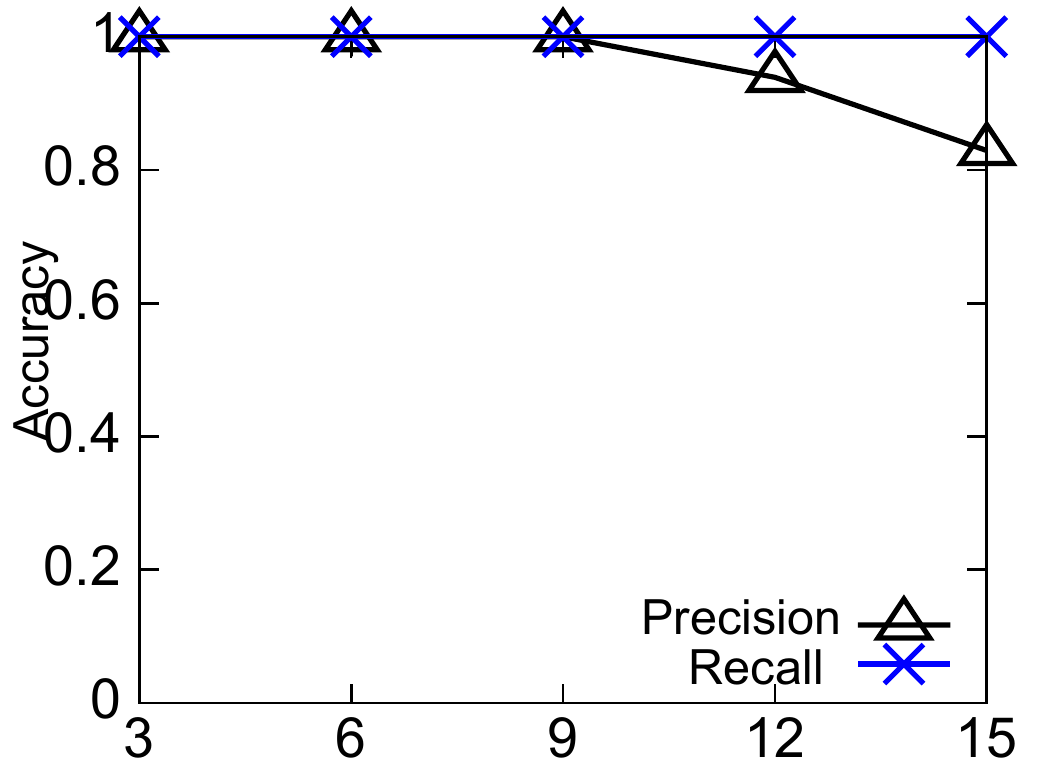}}}
		\hfill\subfloat[
	\small{Time vs. $err\%$ }]
	    {\label{fig:timeerr}
			{\includegraphics[width=3.2cm,height=2.8cm]{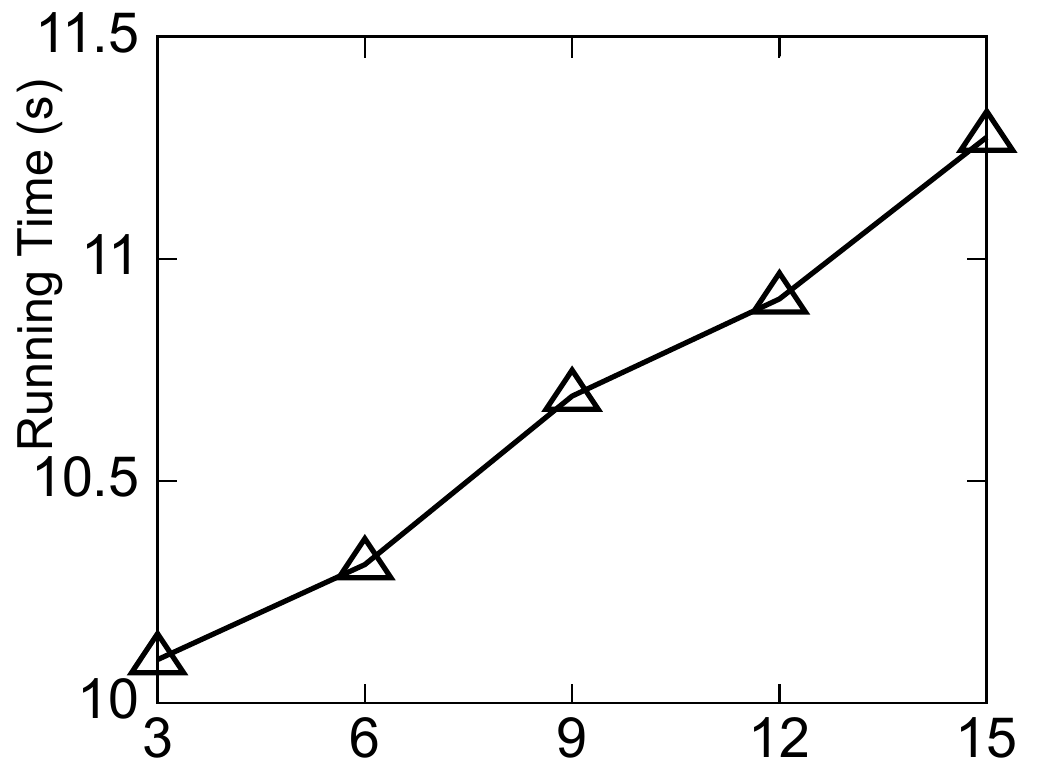}}}
	\vspace{-2ex}
	\caption{Sense assignment effectiveness
	(Clinical dataset). }\label{fig-efficiency}
\end{figure*}

\subsection{OFDClean Performance}

\noindent\textbf{Exp-9: Vary beam size $b$.}
Figure~\ref{fig:bsacckiva} shows that as we evaluate more candidate ontology repairs, \eat{we increase the likelihood of finding better repairs,} we achieve higher precision and recall.  The incremental benefit becomes marginal once we have found the best repair, as reflected between $b = 4$ and $b = 5$.  Figure~\ref{fig:bsrtkiva} shows that the runtime increases exponentially due to the number of repair combinations that must be evaluated for increasing $b$.  To manage runtime costs in practice, the beam size can be tuned according to accuracy requirements, with initial tuning guidelines given in Section~\ref{sec:beamsearch}. 

\noindent\textbf{Exp-10: Vary $err\%$.}
For increasing error rates, Figure~\ref{fig:acckivaer} shows that accuracy declines due to overlapping values between antecedent and consequent values among multiple OFDs.  An update to attribute $A \in X_2$ for $\phi_1: X_1 \rightarrow A$,  $\phi_2: X_2 \rightarrow B$, changes the distribution of equivalence classes w.r.t. $X_2$, leading to lower recall and precision.
Figure~\ref{fig:rtkivaer} shows that runtimes increase as more errors are evaluated, and a larger space of repairs must be considered. 

\noindent\textbf{Exp-11: Vary $inc\%$.}
Figure~\ref{fig:inclincal} shows the repair accuracy as we vary the incompleteness rate, $inc\%$, 
\begin{wrapfigure}{r}{6.7cm}
\vspace{-4mm}
\begin{center}
\begin{minipage}[b]{0.45\linewidth}
\centering
\includegraphics[width=3.1cm]{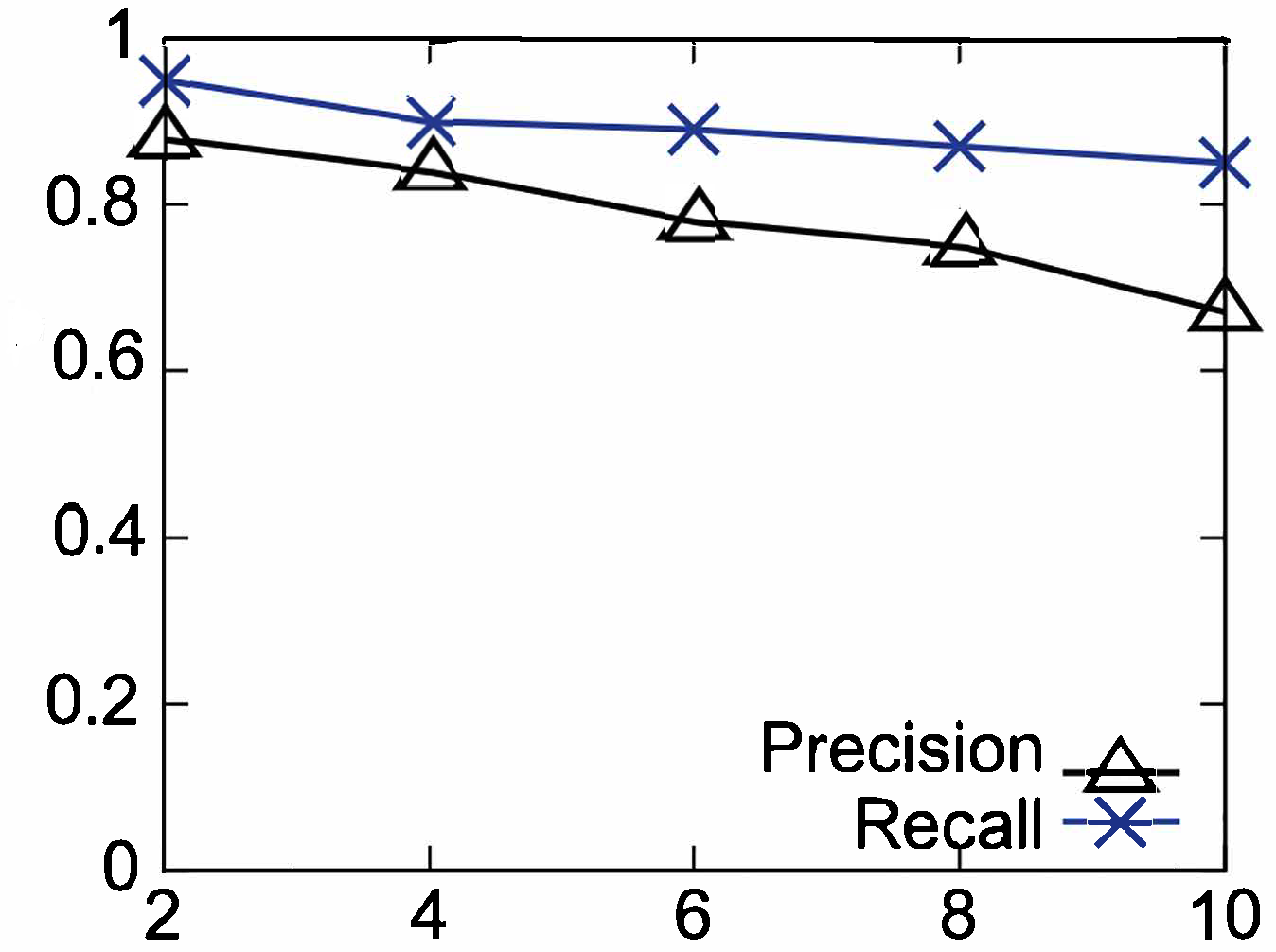}
\caption{Acc. vs. $inc\%$}
\label{fig:inclincal}
\end{minipage}
\hspace{-0.1cm}
\begin{minipage}[b]{0.45\linewidth}
\centering
\includegraphics[width=3.1cm]{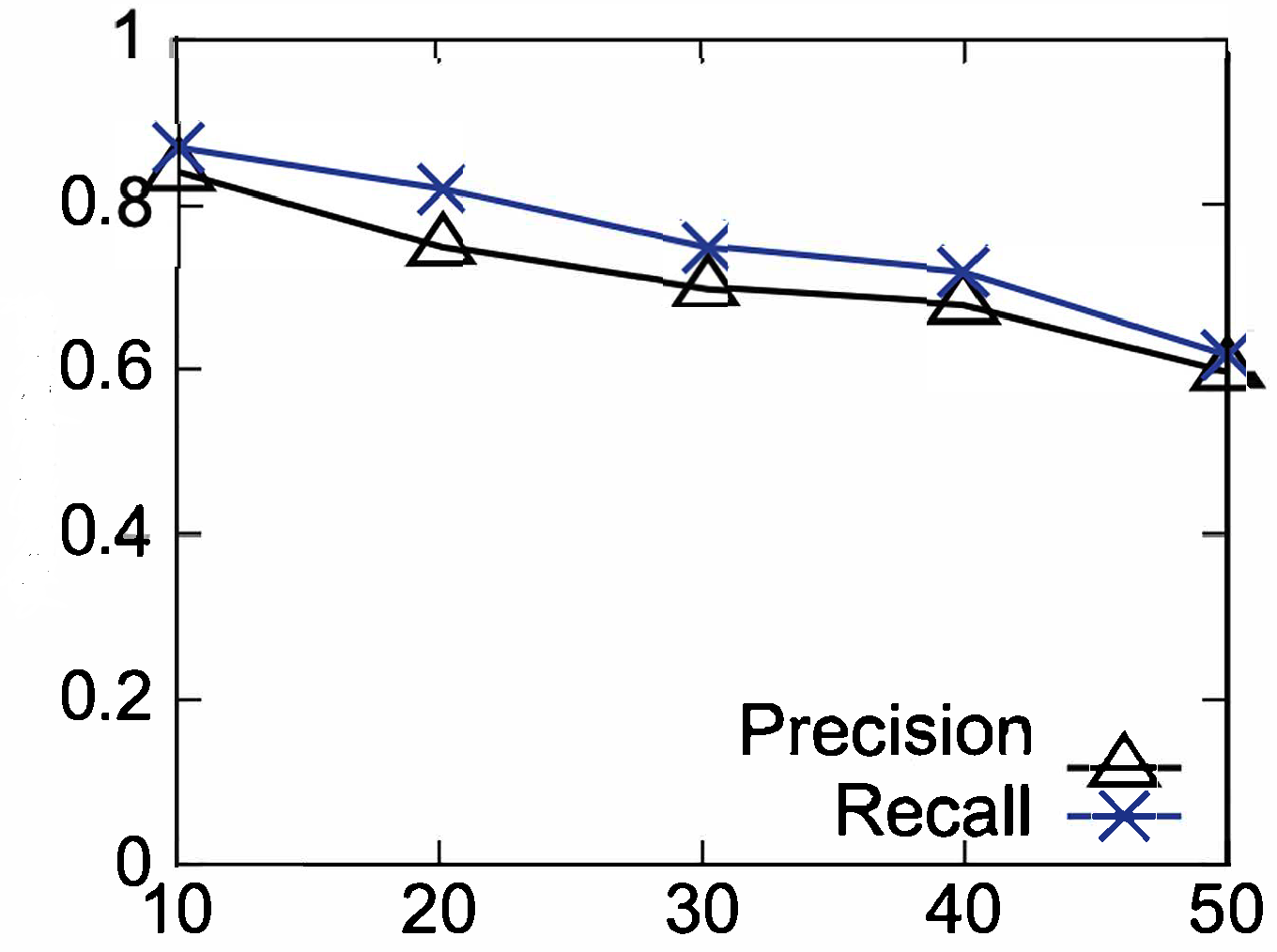}
\caption{Acc. vs. $|\Sigma|$.}
\label{fig:incsigma}
\end{minipage}
  \end{center}
\vspace{-6mm}
\end{wrapfigure}
which measures the percentage of values in $I$, 
but not in $S$, where such errors are resolved via ontology repairs.  As $inc\%$ increases, the precision declines, as some repair values are added to the wrong sense. Recall scores are more consistent, achieving above 85\%, with slight linear decline as error values are corrected by data updates. 

\noindent \textbf{Exp-12: Vary $|\Sigma|$.}  Figure~\ref{fig:incsigma} shows that increasing the number of OFDs causes both precision and recall to decline as an increasing number of attributes overlap among the OFDs.  
\blue{\ofdclean resolves errors when attribute overlap occurs in the consequent attribute between two OFDs, i.e., $\phi_1: X_1 \rightarrow A_1$,  $\phi_2: X_2 \rightarrow A_2$, $A_1 = A_2$.  However, when errors arise w.r.t. $\phi_1$ due to an update in attribute $A_2 \in X_1$, we do not re-evaluate $\phi_1$ again.  We intend to explore this increased space of repairs in future work.}


\begin{figure*}[tb!]
	\captionsetup[subfloat]{justification=centering}
	\centering
	\subfloat[\small{Accuracy vs. $b$ (Kiva).
		}]{\label{fig:bsacckiva}
			{\includegraphics[width=3.2cm,height=2.8cm]{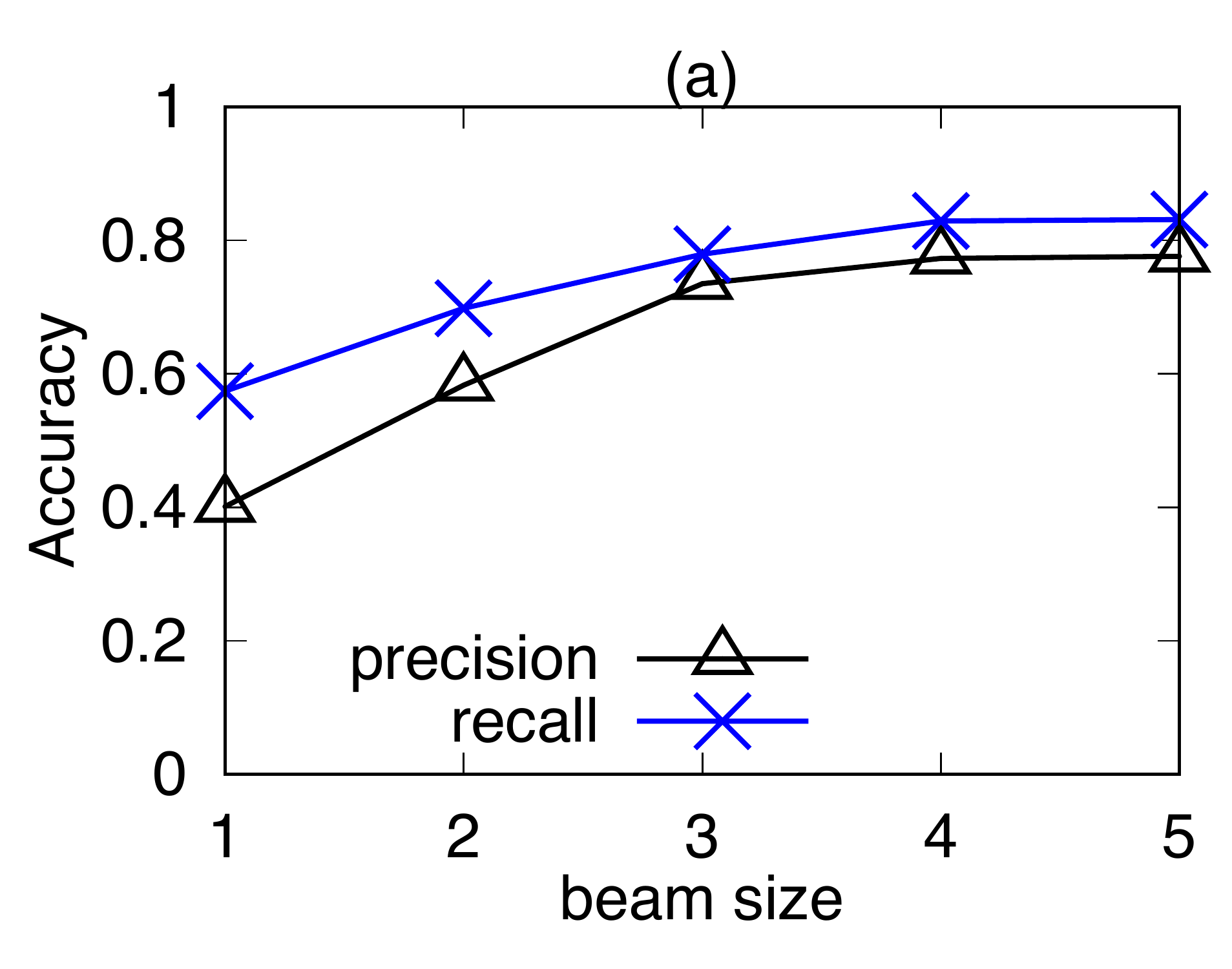}}}
		\hfill\subfloat[
		\small{Runtime vs. $b$ (Kiva).} 
		]{\label{fig:bsrtkiva}
			{\includegraphics[width=3.2cm,height=2.8cm]{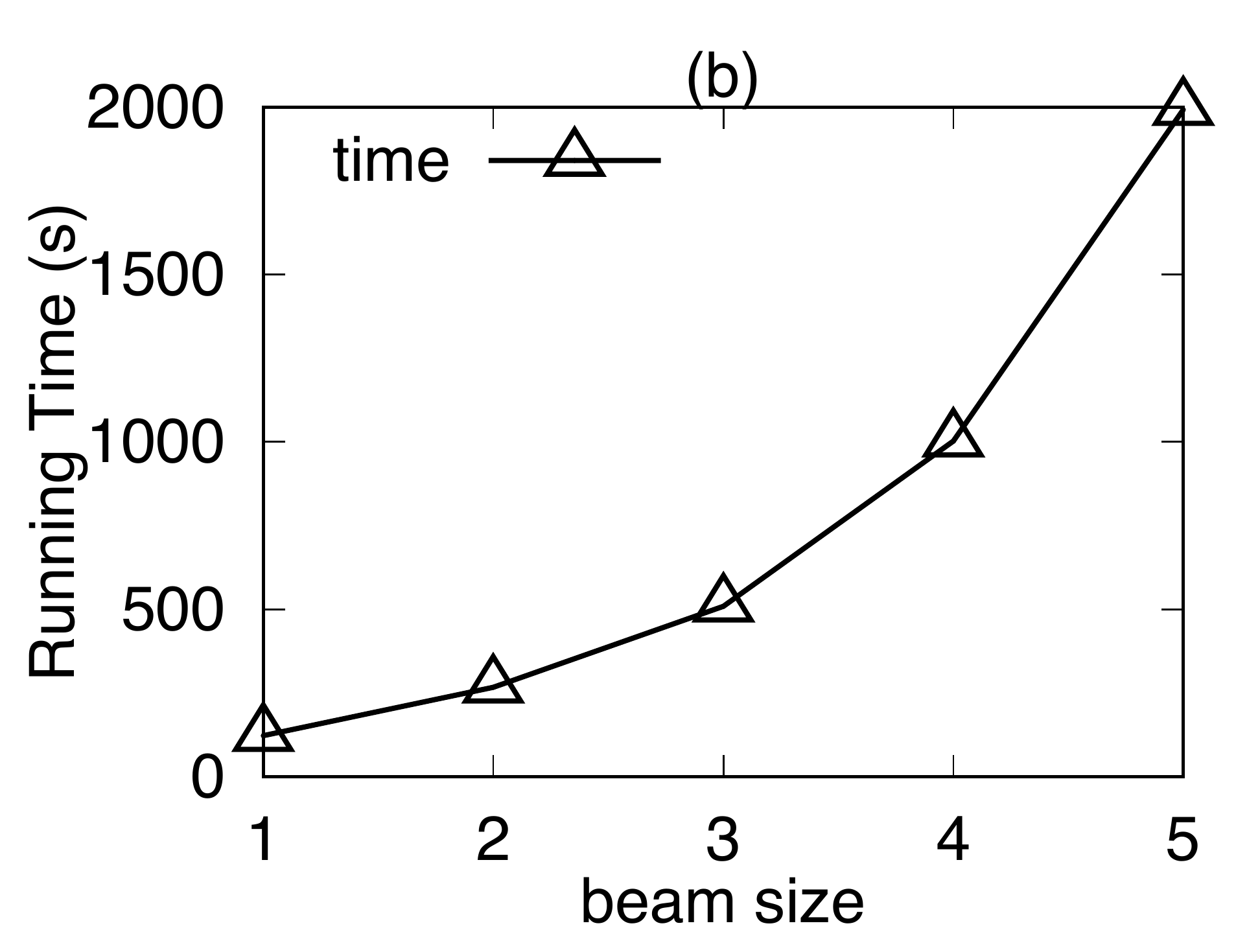}}}
        \hfill \subfloat[
		\small{Acc. vs. err\% (Kiva).}]{\label{fig:acckivaer}
			{\includegraphics[width=3.2cm,height=2.8cm]{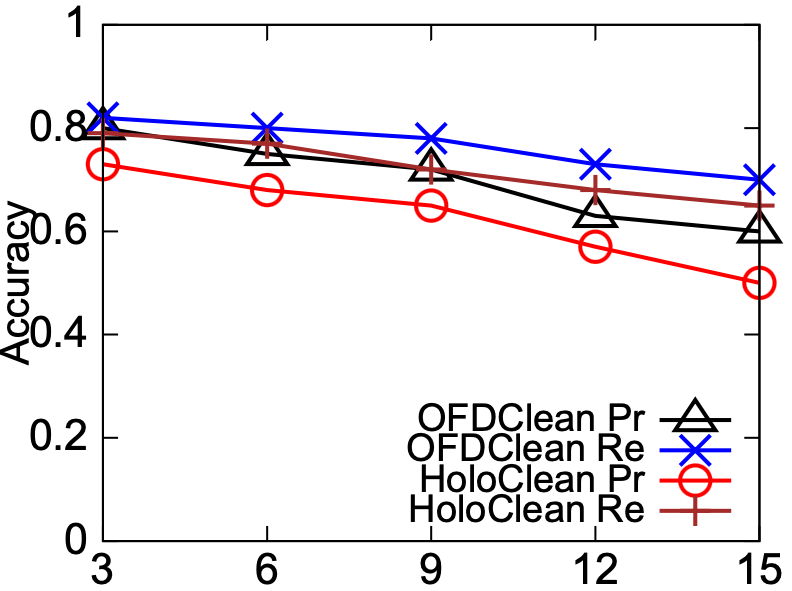}}}
		\hfill\subfloat[\small{Time vs. err\% (Kiva).}]
		{\label{fig:rtkivaer}
			{\includegraphics[width=3.2cm,height=2.8cm]{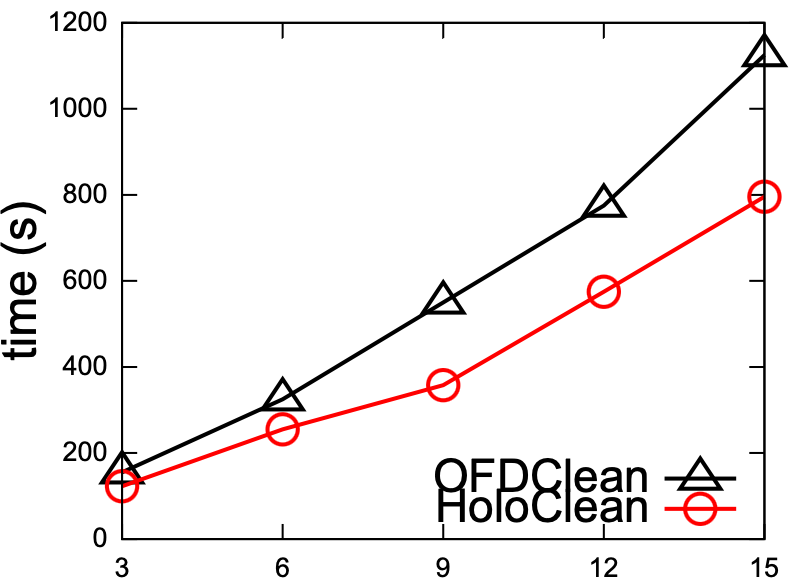}}}	
	\vspace{-2ex}
	\caption{OFDClean Performance. }\label{fig-efficiency}
	\vspace{-0.35cm}
\end{figure*}

\noindent \textbf{Exp-13: Vary $N$.}
Increasing the number of tuples, $N$, does not significantly impact the repair accuracy, with an average variation of 1.4\% in precision (figure omitted for brevity).  Similarly, \ofdclean \ achieves linear runtime performance, as shown in Table~\ref{tab:runtime}.  
This runtime increase occurs due to the larger number of equivalance classes created by newly inserted domain values, and the overhead of evaluating ontology repair combinations. 

\noindent \textbf{Exp-14: Comparative Discussion.}
To the best of our knowledge, there are no ontology repair algorithms coupled with data repair, especially for sense selection.  We compare \ofdclean \ against HoloClean~\cite{holoclean}, which considers external information during repair selection.  Figure~\ref{fig:acckivaer} shows that \ofdclean \ achieves increased precision and recall over HoloClean by 7.4\% and 4.4\%, respectively.  By recognizing the same interpretation of different values via a sense,  \ofdclean \  selectively identifies errors and improves precision.  This comes at a cost in runtime (Figure~\ref{fig:rtkivaer}) since \ofdclean must evaluate a space of ontology repairs to find those that  minimize the number of data repairs.   


\section{Related Work}
\label{sec:rw}

\textbf{Data Dependencies.}  Our work is similar to ontological dependencies and relaxed notions of FDs.  
\eat{OFDs combine ontologies with data dependencies to enforce data integrity.  Earlier techniques have proposed FDs over RDF triples based on the co-occurrence of values.  However, these FDs do not consider structural requirements to specify which entities should carry the values \cite{ACP10,HGPW16}.}   Motik et al.\ define OWL-based integrity constraints where the OWL ontologies are incomplete~\cite{MHS07}.  They study how inclusion dependencies and domain constraints can be used to check for missing and valid domain values within an ontology.  \eat{The proposed constraints do not model functional dependencies (as considered in this work) since the focus is on data completeness.}   
\emph{Ontological Graph Keys (OGKs)} incorporate an event pattern defined on an entity to identify instances w.r.t. concept similarity in an ontology~\cite{MAWC19}.  OGKs characterize entity equivalence using ontological similarity.   Our work shares a similar spirit to leverage ontological relationships and concepts. However, existing techniques do not consider the notion of senses to enable multiple interpretations of values in an ontology.  


Past work has proposed relaxed notions of FDs including strong and weak FDs~\cite{LL98}, and null FDs (NFDs) defined over incomplete relations (containing null values)~\cite{Lien82}.  A weak FD holds over a relation if there exists a possible world (instance) when an unknown value is updated to a non-null value.  \eat{Strong FDs hold when the dependency holds for all possible worlds when an unknown value is updated to a non-null.} OFDs are similar to weak FDs when consequent values exist in the data but not in the ontology.   We have shown that despite having equivalent axiom systems, the semantics of OFDs differ from NFDs, and verification of OFDs must be done over equivalence classes w.r.t. left-hand-side attributes of the OFD versus pairs of tuples for NFDs.

While OFD and NFD discovery can be modeled via a set-containment lattice, OFD discovery validates, for each equivalence class, whether there is a non-empty intersection of the senses.  This incurs increased complexity during verification based on the number of synonyms and senses in the ontology.  Computing data repairs to achieve consistency w.r.t. OFDs requires finding a sense that minimizes a cost function, e.g., the number of updates to achieve consistency, where interactions between OFDs, and finding the best sense, both incur additional challenges beyond NFDs.  \reviseOne{Recent work has shown that the parameterized complexity of detecting errors w.r.t. FDs with left-hand size at most $k$ is W[2]-complete~\cite{BFS22}.  An interesting avenue of future work is to explore whether this complexity extends to OFD detection.}



\noindent \textbf{Constraint-based Data Cleaning.}  
FDs have served as a benchmark to propose data repairs such that the data and the FDs are consistent \cite{BFFR05,PSC15,CIP13a}, with recent extensions to limit disclosure due to data privacy requirements~\cite{CG18,HMC18}.   Relaxed notions of equality in FDs have been proposed by using similarity and matching functions to identify errors, and propose repairs \cite{BKL11,HTL+17}.  While our work is in similar spirit,  we differ in the following ways: (1) the similarity functions match values based only on syntactic string similarity; \eat{using functions such as edit-distance, Jaccard, and Euclidean distance.  Hence, values such as  'India' and 'Bharat' would be dissimilar;} and (2) our cleaning algorithms directly use the notion of senses to enable  similarity under multiple interpretations.  
\reviseOne{Recent work has studied the complexity of computing optimal subset repairs with tuple deletions and updates w.r.t. FDs~\cite{LKR20,MCLG20}.  This  includes a polynomial-time algorithm over a subclass of FDs with tuple deletion~\cite{LKR20}, and approximating optimal repairs within a constant factor less than 2~\cite{MCLG20}.  Studying how these results can be applied to optimal repairs w.r.t. synonym OFDs and tuples updates is an interesting next step.}

\noindent \textbf{Holistic Data Cleaning.}
Probabilistic approaches have shown promise to learn from given attribute templates and training samples to infer clean values~\cite{ holoclean}. Existing systems either use a broad set of constraints~\cite{GMPS13}, or bound repairs according to maximum likelihood~\cite{YBE13}.  \eat{The LLUNATIC framework supports more holistic data cleaning by using a general form of equality generating dependencies (EGDs) covering a broad set of dependencies~\cite{GMPS13}.  The SCARED approach learns from clean portions of the data to find bounded repairs according to maximum likelihood~\cite{YBE13}.}  These techniques capture context either via external sources or additional statistics.  We advocate for a deeper integration of context into the data cleaning framework so that ``inconsistencies'' are not flagged in the first place.  KATARA~\cite{CMI+15} includes simple patterns from ontologies such as ``France'' hasCapital ``Paris'', but does not integrate ontologies into the definition of integrity constraints.  Context takes a central role in BARAN, which defines a set of error corrector models with a context-aware data representation to improve precision~\cite{MahdaviA20}.  Incorporating these models into \ofdclean and resolving conflicts between contextual models and ontologies are interesting avenues of future work.

\noindent \textbf{Knowledge Base Incompleteness.}
Computational fact checking against Knowledge Bases (KBs) has emerged in recent years to automatically verify facts from different domains.  The main problem in fact checking with KBs is that the reference information may be incomplete.  That is, under an Open World Assumption (OWA), we cannot assume that all facts in a KB are complete; facts not in a KB may be false or just missing~\cite{HP18}.  In our work, we do not make such open assumptions.  To improve KB quality, the RuDiK system discovers positive rules to mitigate incompleteness (e.g., if two persons have the same parent, then they are siblings), and negative rules to detect errors (e.g., if two persons are married, one cannot be a child of the other)~\cite{OMP18}.  Incorporating such inference into \ofdclean as a new ontology repair operation will enrich existing ontological relationships.

\section{Conclusions}
\label{sec:conclusion}

Ontology Functional Dependencies (OFDs) capture domain relationships found in ontologies.  We studied fundamental problems for OFDs, and proposed the \fastofd algorithm to discover a minimal set of OFDs.  We also studied the repair problem with respect to a set of OFDs, and proposed an algorithm to find a sense for each equivalence class to provide an interpretation for each OFD against the data.  Our framework, \ofdclean, combines sense selection with ontology and data repair such that the total number of updates to the data and to the ontology is minimized.  Our experiments showed that our algorithms are effective, as OFDs are useful data quality rules to capture domain relationships, and significantly reduce the number of false positive errors in data cleaning solutions that rely on traditional FDs.   \reviseOne{As next steps, we intend to study synonym relationships in antecedent attributes}, and explore cleaning with respect to inheritance OFDs, which represent hierarchical relationships among data values.




\bibliographystyle{ACM-Reference-Format}
\bibliography{ref}



\end{document}